\documentclass[a4paper,11pt]{article}
\usepackage{jheppub} 
\usepackage{amssymb}
\usepackage{amsthm}
\usepackage{amstext}
\usepackage{amsmath}
\usepackage{breqn}
\setkeys{breqn}{breakdepth={0}}
\usepackage{amsfonts}
\usepackage{bbm}
\usepackage[german,english]{babel}

\usepackage{enumitem}
\usepackage{algorithm}
\usepackage{algorithmic}

\theoremstyle{plain}
\newtheorem{thm}{Theorem}[section]
\newtheorem{cor}[thm]{Corollary}
\newtheorem{lem}[thm]{Lemma}
\theoremstyle{definition}
\newtheorem{defn}[thm]{Definition}
\newtheorem{exmp}[thm]{Example}
\newtheorem{rem}[thm]{Remark}
\preprint{MPP-2020-139, PCFT-20-26, USTC-ICTS-20-26}

\title{\boldmath IBP reduction coefficients made simple}

\author[a]{Janko Boehm}
\author[a]{Marcel Wittmann}
\author[b,c]{Zihao Wu}
\author[d]{Yingxuan Xu}
\author[b,c,e]{Yang Zhang}

\affiliation[a]{Department of Mathematics, Technische Universit\"at Kaiserslautern, 67663 Kaiserslautern, Germany}
\affiliation[b]{Peng Huanwu Center for Fundamental Theory, Hefei, Anhui 230026, China}
\affiliation[c]{Interdisciplinary Center for Theoretical Study, University of Science and Technology of China, Hefei, Anhui 230026, China}
\affiliation[d]{ETH Z\"urich, Institut f\"ur Theoretische Physik,
Wolfgang-Pauli-Str. 27, 8093 Z\"urich, Switzerland}
\affiliation[e]{Max-Planck-Institut f\"ur Physik,
  Werner-Heisenberg-Institut, D-80805, M\"unchen, Germany}

\emailAdd{boehm@mathematik.uni-kl.de}
\emailAdd{mwittman@rhrk.uni-kl.de}
\emailAdd{wuzihao@mail.ustc.edu.cn}
\emailAdd{yingxu@student.ethz.ch}
\emailAdd{yzhphy@ustc.edu.cn}

\abstract{We present an efficient method to shorten the
analytic integration-by-parts (IBP) reduction coefficients of multi-loop Feynman integrals. For our approach, we develop an
improved version of Leinartas' multivariate partial fraction
algorithm, and provide a modern implementation based on the computer
algebra system \texttt{Singular}. Furthermore, We observe
that for an integral basis with uniform transcendental (UT) weights,
the denominators of IBP reduction coefficients with respect to the UT
basis are either symbol letters or polynomials purely in the spacetime
dimension $D$. With a UT basis, the partial fraction algorithm is  more efficient both with respect to its performance and the size reduction. We show that in complicated examples with
existence of a UT basis, the IBP reduction coefficients size can be
reduced by a factor of as large as $\sim 100$. We observe that our algorithm also works well for settings without a UT basis.}

\begin{document} 

\maketitle
\flushbottom

\section{Introduction}
\label{sec:intro}
With the end of Large Hadron Collider (LHC) run-II and the upgrade to
HL-LHC \cite{ApollinariG.:2017ojx, Abada:2019ono}, there is an eager demand for high-precision physics
computations. The computation of integration-by-parts
(IBP) identities~\cite{Tkachov:1981wb,Chetyrkin:1981qh}, which can be used to reduce a large number of Feynman integrals to a small set of master integrals, is a critical and often bottleneck step for the evaluation of multi-loop scattering amplitudes
in precision physics.  

There are many publicly available IBP reduction
programs, like {\sc AIR}, {\sc FIRE}, {\sc Kira}, {\sc Reduze} and {\sc LiteRed}~ \cite{Anastasiou:2004vj,Smirnov:2008iw,Smirnov:2013dia,Smirnov:2014hma,Smirnov:2019qkx,
  Maierhoefer:2017hyi,Maierhofer:2018gpa,Maierhofer:2019goc,Studerus:2009ye,vonManteuffel:2012np,Lee:2013mka,Klappert:2020nbg}, based on
the Laporta algorithm~\cite{Laporta:2001dd} and the algebra structures
of IBP relations~\cite{Smirnov:2006wh,Smirnov:2006tz,Lee:2008tj}. In recent years, many new ideas and programs have appeared
for use in the computation of complicated multi-loop IBP reductions,
for example, syzygy
approach~\cite{Gluza:2010ws,Schabinger:2011dz,Ita:2015tya,Larsen:2015ped,Boehm:2017wjc,vonManteuffel:2020vjv},
finite-field
interpolation~\cite{vonManteuffel:2014ixa,Peraro:2016wsq,Klappert:2019emp,Klappert:2020aqs,Peraro:2019svx}, module intersection~\cite{Boehm:2018fpv,Bendle:2019csk}, intersection theory~\cite{Mastrolia:2018uzb,Frellesvig:2019uqt,Frellesvig:2019kgj,Frellesvig:2020qot}, $\eta$ expansion~\cite{Liu:2017jxz,Liu:2018dmc,Guan:2019bcx,Zhang:2018mlo,Wang:2019mnn}   and direct solution of IBP recursive relations~\cite{Kosower:2018obg}. 

Besides the development of the computational techniques for IBP
reductions, there is another problem which was less addressed in the
literature. 
Frequently, after an analytic IBP reduction of complicated multi-loop Feynman
integrals, we obtain reduction coefficients with a huge size, as rational functions of the
spacetime parameter $D$ and kinematic variables. The huge
coefficients are difficult to store, to transfer, to use for analytic scattering amplitude
computations, and also very cumbersome for numerical evaluations. 
Thus, an important question arises:
\begin{quotation}
  {\it How do we 
  simplify the analytic IBP reduction coefficients in practice?}
\end{quotation}

One natural idea to make analytic IBP reduction coefficients shorter,
is to choose a ``good''  master integral basis. Early attempts were
made to test different integral orderings in the Laporta algorithm, in
order to get shorter reduction coefficients. However, it is difficult
to dramatically shorten IBP reduction coefficients by simply changing
the integral ordering. Recently, new methods were
presented~\cite{Smirnov:2020quc, Usovitsch:2020jrk} to find a good
master integral basis such that the dimensional parameter $D$
factorizes out in the final IBP reduction coefficients and makes the
reduction much easier.
In ref.~ \cite{Bendle:2019csk}, the master integral basis with uniform
transcendental (UT) weights ~\cite{Henn:2013pwa,Henn:2014qga} was suggested to shorten the size of IBP
reduction coefficients.

In this paper, we propose a powerful method to reduce the byte size of
the analytic IBP reduction coefficients, which is based
on our modern version of Leinartas' multivariate partial
fraction algorithm~\cite{leinartas1978factorization,
  raichev2012leinartas}.  Leinartas' algorithm has been used for
solving basis transformation matrix in Meyer's UT determination
algorithm \cite{Meyer:2017joq}, and for the reconstruction and
simplification of the planar two-loop five-parton pentagon function
coefficients \cite{Abreu:2019odu}. We
develop an improved version of Leinartas' algorithm and implement it in a library for the open source computer algebra system {\sc Singular} \cite{DGPS}. From the examples
we have tested, this method can rewrite a huge rational function in 
IBP reduction coefficients as a much shorter sum of simpler rational
functions. 

The improvements to Leinartas' algorithm include an additional decomposition
step between the first step (Nullstellensatz decomposition) and second step (algebraic dependence decomposition) of the algorithm which reduces the size of the denominators (and numerators) by doing a (multivariate) division
with remainder by the denominator factors. Moreover, in addition to Leinartas' original algorithm, we add a third step in the algorithm, which implements a numerator decomposition as suggested in \cite{Meyer:2017joq} and uses a syzygy computation to reduce the size of the decomposition expression. 
In particular in the case of examples arising from IBP reductions, due to the additional decomposition step and by reducing the size of the algebraic relations used, we were able to drastically reduce  the runtime of the second step of Leinartas' algorithm, which relies on algebraic relations between the denominator factors.
For this we make use of \textsc{Singular}'s efficient algorithms for calculating Gr\"obner
bases, syzygy modules and polynomial factorizations. We provide a detailed description of
the algorithm in pseudocode.

As an algorithm based on partial fractioning, the size reduction
ratio and the running time depend on the degree of irreducible
denominators. We combine our partial fractioning approach with the
strategy of choosing a ``good'' master integral basis. In particular, as
mentioned in ref.~\cite{Bendle:2019csk}, we suggest that when
a UT master integral basis for the integral family under consideration
exists, it is
advantageous to first reduce Feynman integrals to the UT basis, and
then run our partial fraction algorithm to shorten the size of the IBP
coefficients. The reason is that, in the examples we have tested, for
Feynman integrals 
\begin{gather}
G[\alpha_1,\ldots,\alpha_j]=\int\prod_{j=1}^L \frac{d^D l_j}{i
  \pi^{D/2}}\frac{1}{\prod_{i=1}^n D_i^{\alpha_i}},\quad \alpha_i \in
\mathbb Z
\end{gather}
with each $D_i$ defined as a square of a $\mathbb Z$-linear combination of loop
and external momenta minus the mass term, IBP reduction coefficients with respect to a UT basis have the following good properties:
\begin{itemize}
\item The spacetime dimension parameter $D$ factorizes out in the denominator
  of the reduction coefficients.
\item Except the 
  factors purely in $D$'s, the other factors in the IBP reduction
  coefficients' denominators, are (a subset of) the symbol letters.
\end{itemize}
Therefore, using a UT basis, we usually get much simpler irreducible
factors in the denominators of IBP reduction coefficients. This property
makes the partial fractioning much faster and the result usually shorter
than that from the usual master integral choice. 

We tested various IBP reduction coefficients from simple diagrams to
complicated frontier diagrams. In some complicated IBP reduction coefficients examples, we observe that
our partial fractioning algorithm, combined with the UT basis choice, dramatically shortens the coefficient size by a factor of as
large as $~100$.  In the Appendix \ref{An explicit example of the size reduction}, we explicitly list
an example of one coefficient, before and
after the partial fraction decomposition to provide an impression of this dramatic reduction of size.

We distribute the {\sc Singular} code of our partial fraction
implementation as an open source {\sc Singular} library for download:
\begin{quotation}
\centerline{\url{https://github.com/Singular/Singular/tree/spielwiese/Singular/LIB/pfd.lib}}
\end{quotation}

This paper is organized as follows: In Section  \ref{sec:IBP_MI} we
set up the notations and review the concepts of IBP reduction and master integrals. In
Section \ref{sec:partial}, we present our improved verion of Leinartas'
algorithm to shorten IBP reduction coefficients. In Section
\ref{sec:example}, we provide several IBP reduction simplifications,
and also emphasize the benefit of using UT bases in case they exist. In
Section \ref{sec:summary}, we summarize our discoveries and discuss
possible directions for future research. In the appendices, we provide
a manual describing the use of our {\sc Singular} library for
multivariate partial fractioning, and an explicit example of  the coefficient size reduction.

\section{IBP and Master Integrals}
\label{sec:IBP_MI}

\subsection{Integration-by-Parts Identities and master integrals}

There are many algebraic relations between different Feynman integrals and it is very efficient to use these relations to obtain further Feynman integrals from the ones we already know. A very useful set of relations can be obtained via the integration-by-parts (IBP) identities, which relate different integrals of a given integral family.

Consider a Feynman integral with any loops
\begin{gather}
\int\prod_{j=1}^L \frac{d^D l_j}{i \pi^{D/2}}\frac{1}{\prod_{i=1}^n D_i^{\alpha_i}},
\end{gather}
where $L$ is the number of loops, $\alpha_i$ are integer indices and the denominators are given by 
\begin{gather}
D_i=\sum_{j\geq k\geq 1}^L A_i^{jk} l_j\cdot l_k+\sum_{i=1}^L B_i^j \cdot l_j+E_i,
\end{gather}
i.e are quadratic or linear functions of the external momenta $p_i$
and the loop momenta $l_i$.

The standard IBP relation ~\cite{Tkachov:1981wb,Chetyrkin:1981qh} is,
\begin{gather}
0=\int\prod_{j=1}^L \frac{d^D l_j}{i \pi^{D/2}} \frac{\partial}{\partial l_m}\left(q_k \prod_{i=1}^n D_i^{-\alpha_i}\right),
\label{IBP-relation}
\end{gather}
where $m=1, \dots, L$ with $q_k$ a linear combination of loop momenta and external momenta.

With the IBP identities, we can find the basis of a given integral family, which are called master integrals (MIs). The finiteness of master integrals was proven in ref.~\cite{Smirnov:2010hn}. 

So a Feynman integral can be written as a linear combination of master integrals,
\begin{gather}
I[\alpha_1, \dots, \alpha_n]=\sum_{i} c_i I_i,
\end{gather}
here $\alpha_i$ are integer indices of denominators and $I_i$ are master integrals.

In practice, IBP reduction can be done by many algorithms, such as the Laporta algorithm~\cite{Laporta:2001dd}, the algebra structures
of IBP relations~\cite{Smirnov:2006wh,Smirnov:2006tz,Lee:2008tj}, finite-field
interpolation~\cite{vonManteuffel:2014ixa,Peraro:2016wsq,Klappert:2019emp,Klappert:2020aqs,Peraro:2019svx}, module instersection~\cite{Boehm:2018fpv}, intersection theory~\cite{Mastrolia:2018uzb}, $\eta$ expansion~\cite{Liu:2018dmc}  and direct solution of IBP
recursive relations~\cite{Kosower:2018obg}. And there are also many  public IBP reduction
codes, like {\sc AIR}, {\sc FIRE}, {\sc Kira}, {\sc Reduze} , {\sc LiteRed} \cite{Anastasiou:2004vj,Smirnov:2008iw,Smirnov:2013dia,Smirnov:2014hma,Smirnov:2019qkx,
  Maierhoefer:2017hyi,Maierhofer:2018gpa,Maierhofer:2019goc,Studerus:2009ye,vonManteuffel:2012np,Lee:2013mka}.

\subsection{Differential equation and UT basis}
\label{DE-UT}

Since the master integrals are functions of scalar products of external momenta, it is natural to consider the derivatives with respect to the scalar products. By introducing a vector
\begin{gather}
    \vec{I}=\left(
\begin{array}{c}
 I_1 \\
 I_2 \\
 \dots\\
 I_n \\
\end{array}
\right),
\end{gather}
here $I_i$ are the master integrals of a corresponding Feynman diagram, we can set up the following differential equation
\begin{gather}
d \vec{I}=(d A) \vec{I},
\end{gather}
where $A$ is a $n\times n$ matrix. Normally, every element of $A$ is a rational function of spacetime dimension $D$ and kinematic variables. 

While, Johannes Henn showed that with a new choice of MIs, differential equations can simplify in a way that they can be solved easily order by order~\cite{Henn:2013pwa,Henn:2014qga}.  With suitable MIs,  the differential equation can be written like that
\begin{gather}
    d\vec{I'}=\epsilon (d A) \vec{I'}
\label{DE}
\end{gather}
with 
\begin{gather}
    A=\sum A_k \log S_k.
\label{dlog}
\end{gather}
This is called the canonical form of differential equations, here we set $D=4-2\epsilon$ and each $A_k$ is a constant matrix, $S_k$ are functions of Lorentz invariants, which are called symbol letters.

With~\eqref{DE}, the differential equations can be solved order by order in an $\epsilon$-order expansion:
\begin{gather}
\label{UT-expansion}
    \vec{I'}=\vec{I'_0}+\epsilon \vec{I'_1}+ \epsilon^2 \vec{I'_2}+\dots,\\\nonumber
    d \vec{I'_1}=(d A) \vec{I'_0}, \quad d \vec{I'_2}= (d A) \vec{I'_1} \dots.
\end{gather}

The key property of these suitable master integrals can be described with the concept of the degree of transcendentality $\mathcal{T} (f)$  of a function. $\mathcal{T} (f)$ defines the fold number of iterated integrals needed in the function $f$.  Moreover, we require $\mathcal{T}(f_1 f_2)=\mathcal{T}(f_1)+\mathcal{T}(f_2)$. So that, we can see 
\begin{gather}
\mathcal{T}(\mathrm{Li}_k(x))=k,\quad  \mathcal{T}(\log x)=1, \quad \mathcal{T}(\zeta_n)=\mathcal{T}(\mathrm{Li}_n(1))=n,\\\nonumber
    \mathcal{T}(algebraic \; factors)=0,\quad \mathcal{T}(\zeta_2)=\mathcal{T}(\frac{\pi^2}{6})=2\Rightarrow\mathcal{T}(\pi)=1.
\end{gather}

If the function also satisfies 
\begin{gather}
    \mathcal{T}\left(\frac{d}{dx}f(x)\right)=\mathcal{T}(f(x))-1,
\end{gather}
then the function $f$ is called a pure function. With this definition we can see that
if we multiply a pure function with an algebraic function of $x$,  the resulting function would still have the same uniform transcendentality but no longer be a pure function anymore, since the derivative is also applied on the algebraic function. 

Because of \eqref{DE} and \eqref{UT-expansion}, we can see that the functions in $\vec{I'_k}$ are all pure functions, hence the $\vec{I'}$ is called uniform transcendental (UT) basis. 

There are many ways to construct a UT basis. For examples, we can
construct it via Fuchsia and epsilon, based on the Lee's
algorithm~\cite{Gituliar:2016vfa,Prausa:2017ltv,Lee:2014ioa}. Meyer
proposed a package CANONICA to find a transformation to get UT
integrals~\cite{Meyer:2017joq}. What is more, by means of leading
singularity analysis and the dlog ansatz, we can also construct a UT
basis~\cite{Wasser:2018qvj}. A UT basis can be also constructed via
Baikov analysis~\cite{Chicherin:2018old}, and systematically via the
dlog form in a general representation and the intersection theory
\cite{Chen:2020uyk}.  And recently, it was discovered that the full UT basis from only one UT integral~\cite{Dlapa:2020cwj}.

\subsection{Symbol of a transcendental function}

In section~\ref{DE-UT}, we proposed the canonical form of differential equation
\begin{gather}
    d\vec{I'}=\epsilon (d A) \vec{I'}
\end{gather}
with 
\begin{gather}
    A=\sum A_k \log S_k.
\end{gather}

In the case where the symbol letter alphabet can be written in terms of rational functions (in at least one variable), one can write the answer in terms of Goncharov polylogarithms (also called hyperlogarithms, multiple logarithms)~\cite{Goncharov:1998kja,Henn:2014qga}. The Goncharov polylogarithms can be defined iteratively as follows,
\begin{gather}
    G(a_1,\dots,a_n;z)=\int_{0}^z \frac{dt}{t-a_1} G(a_2,\dots,a_n;t), \quad a_i\in \mathbb{C},
\end{gather}
with
\begin{gather}
    G(z)\equiv G(;z)=1 .
\end{gather}
In the special case where all the $a_i$ are zero, we define, using the obvious vector notation $\vec{a}_n= (a, . . . , a)$, $a \in \mathbb{C}$,
\begin{gather}
    G(\vec{0}_n;z)=\frac{1}{n!}\log^n z, \quad G(\vec{a}_n;z)=\frac{1}{n!}\log^n\left(1-\frac{z}{a}\right).
\end{gather}

A Goncharov polylogarithm $T_k$ of transcendentality degree $k$ can be written as a linear combination (with rational coefficients) of $k$-fold iterated integrals of the form~\cite{Goncharov:2010jf}
\begin{gather}
    T_k =\int_{a}^b d\log R_1\circ\dots \circ d\log R_k,
\end{gather}
where $a$ and $b$ are rational numbers, $R_i(t)$ are rational functions with rational coefficients and the iterated integrals are defined recursively by
\begin{gather}
    \int_{a}^b d\log R_1\circ\dots \circ d\log R_k=\int_{a}^b\left(\int_{a}^t d\log R_1\circ\dots \circ d\log R_{k-1}\right) d \log R_k (t),
\end{gather}
in physics, there $d\log R_k$ are just the ones appeared in eq \eqref{DE}, with $R_k$ equal $S_k$ in \eqref{dlog}.

There is one useful quantity associated with $T_k$ called the symbol, which is an element of the $k$-fold tensor product of rational functions modulo constants~\cite{goncharov2009simple}, denoted by $S$. The symbol of the function $T_k$ is 
\begin{gather}
    \mathrm{symbol} (T_k)\equiv S(T_k)=R_1\otimes R_2\otimes\dots \otimes R_k,
\end{gather}
that is why $S_k$ in \eqref{dlog} are called symbol letters.

There are many other properties of the symbol, see ref.~\cite{Duhr:2011zq, Duhr:2014woa} for a discussion of their properties.

\section{Improved Leinartas' Algorithm and Modern Implementation}
\label{sec:partial}
 
In this section we describe an algorithm based on the work of E.K.Leinartas~\cite{leinartas1978factorization, raichev2012leinartas} to reduce the size of rational functions by writing them as a sum of functions with ``smaller'' numerators and denominators.

The improvements to the algorithm described in the original paper by Leinartas lie mainly in 
an additional decomposition step described in Algorithm \ref{algo:shortNumeratorDecompStep}, which reduces the size of the numerators and denominators by doing a (multivariate) division with remainder
by the denominator factors, as well as the changes discussed in Remark
\ref{rem:simplified_algDep}, which aim at improving the performance of the second decomposition step (Algorithm \ref{algo:algDependDecompStep}) by reducing the size of the required annihilating polynomials. In addition to 
Leinartas' original algorithm we also add, as suggested in \cite{Meyer:2017joq}, a numerator decomposition 
as the final step of the algorithm and use a syzygy module computation to reduce the size of the decomposition (see Algorithm \ref{algo:numeratorDecompStep} and Remark \ref{rem:syzygy_reduction}). 
Thus, while Leinartas' original algorithm calculates a decomposition satisfying only the first two conditions in Theorem \ref{thm:pfd}, we add an additional condition.
In our implementation, we make use of the computer algebra system \textsc{Singular}, which provides efficient algorithms for the calculation of Gr\"obner bases and syzygy modules as well as polynomial factorization.

To state more precisely what we mean by ``smaller'' numerators/denominators, we first need the following definitions. The goal is then an algorithmic proof of Theorem \ref{thm:pfd}. For this, let in the following $K[x_1,\dots,x_d]$ or short $K[\mathbf x]$ be the polynomial ring over some field $K$ in $d$ variables $\mathbf x=(x_1,\dots,x_d)$ and let $\overline K$ denote the algebraic closure of $K$.

\begin{defn} (algebraic dependence)
	A set $\{q_1,\dots,q_m\}\subseteq K[\mathbf x]$ of $m$ polynomials is called \textbf{algebraically dependent} if there exists a nonzero polynomial $p\in K[y_1,\dots,y_m]$ in $m$ variables, such that $p(q_1,\dots,q_m)=0$ in $K[\mathbf x]$. 
	Call $p$ an \textbf{annihilating polynomial} of $q_1,\dots,q_m$.
\end{defn}

\begin{defn}[monomial ordering]
	A \textbf{monomial ordering} for $K[\mathbf x]$ is a total ordering ``$>$'' on the set $\left\{\mathbf x^\alpha\middle|\alpha\in\mathbb N^d\right\}$
	of monomials (writing ``$\mathbf x^\alpha$'' for $x_1^{\alpha_1}\cdot \dots \cdot x_d^{\alpha_d}$), such that $>$ is compatible with multiplication, i.e. for all $\alpha,\beta,\gamma\in\mathbb N^d$ it holds \begin{equation}\mathbf x^\alpha>\mathbf x^\beta\Rightarrow\mathbf x^\alpha\mathbf x^\gamma>\mathbf x^\beta\mathbf x^\gamma\end{equation}
	and $>$ is called \textbf{global} if it is a well ordering or equivalently if $1<x_i$ for all $i=1,\dots,d$. For any polynomial $f\in K[\mathbf x]$ write $L(f)$ for its \textbf{lead monomial}, that is the largest monomial with respect to $>$.
\end{defn}

\begin{defn}[Gr\"obner basis]
	A \textbf{Gr\"obner basis} of an ideal $I\subseteq K[\mathbf x]$ with respect to a given global monomial ordering is a finite subset $G\subseteq I$ such that the ideals generated by all lead monomials of $G$ and of $I$ coincide:\begin{equation}\left<L(g)\middle|g\in G\right>=\left<L(f)\middle|f\in I\right>\end{equation} 
\end{defn}

\begin{defn}[division with remainder]\label{defn:divrem_GB}
	After the choice of a (global) monomial ordering there exists an algorithm (multivariate reduced division with remainder, see \cite[\S3~Theorem 3]{CLO}) to determine for any polynomials $f,g_1,\dots,g_r\in K[\mathbf x]$  a division expression \begin{equation}f=r+\sum_{i=1}^ra_ig_i\qquad(r,a_1,\dots,a_r\in K[\mathbf x])\end{equation} such that none of the lead monomials $L(a_ig_i)$ are bigger than $L(f)$ and no term of $r$ is divisible by any lead monomial $L(g_i)$. Call a polynomial $r$ with this property \textbf{reduced} with respect to $g_1,\dots,g_r$.
	
	In case $G=(g_1,\dots,g_r)$ is a Gr\"obner basis of an ideal $I\subseteq K[\mathbf x]$, it can be shown \cite[\S6~Proposition 1]{CLO}, that the remainder $r$ only depends on the monomial ordering and $I$. In this case call $r$ \textbf{reduced} with respect to $I$.
	
	Call $G$ a \textbf{reduced Gr\"obner basis}, if every $g\in G$ is reduced with respect to $G\backslash\{g\}$. It can be shown \cite[\S7~Proposition 6]{CLO}, that for any ideal $I\subseteq K[\mathbf x]$ a reduced Gr\"obner basis exists and is unique up to multiplication with constants and reordering of the elements. It can be calculated with Buchberger's algorithm \cite[\S7~Theorem 2]{CLO}.
\end{defn}

\begin{thm}[partial fraction decomposition]\label{thm:pfd}
	Let $f,g\in K[\mathbf x]$ and let $g=\prod_{i=1}^mq_i^{e_i}$ be the factorization of $g$ into irreducible factors $(e_i\in\mathbb N)$. Then there exists a decomposition
	\begin{equation}\frac f g = \sum_{S\subseteq\{1,\dots,m\}}\frac{f_S}{\prod_{i\in S}q_i^{b_i}}\qquad (b_i\in\mathbb N,f_S\in K[\mathbf x])\end{equation}
	where all nonzero summands satisfy the following conditions 
	\begin{enumerate}[label=(\arabic*)]
		\item the polynomials $\left\{q_i\middle|i\in S\right\}$ have a common zero in $\overline K^d$
		\item the polynomials $\left\{q_i\middle|i\in S\right\}$ are algebraically independent
		\item $f_S$ is reduced with respect to the ideal $\left<q_i\middle|i\in S\right>\subseteq K[\mathbf x]$
	\end{enumerate}
\end{thm}

Note that $(3)$ depends on the monomial ordering. In order to get numerator polynomials of low degree, a degree ordering (i.e. $\deg(\mathbf x^\alpha)<\deg(\mathbf x^\beta)\Rightarrow \mathbf x^\alpha<\mathbf x^\beta$ for any monomials $\mathbf x^\alpha,\mathbf x^\beta$ in $K[\mathbf x]$) should be chosen. In our \textsc{Singular} implementation we used the graded reverse lexicographic ordering defined by
\begin{align}\label{eq:grevlex_defn}
\nonumber\mathbf x^\alpha>_{grevlex}\mathbf x^\beta :\Leftrightarrow\quad&\deg(\mathbf x^\alpha)>\deg(\mathbf x^\beta)\text{ or } \deg(\mathbf x^\alpha)=\deg(\mathbf x^\beta) \\
&\text{ and the last nonzero entry of } \alpha-\beta \text{ is negative}
\end{align}
Furthermore condition $(2)$ ensures, that at most $d$ different irreducible factors occur in each denominator of the decomposition, since it can be shown, that any set of at least $d+1$ polynomials (in $d$ variables) is algebraically dependent. (This follows directly from the Jacobian criterion \ref{lem:jac_crit}.)

In view of condition $(1)$ in Theorem \ref{thm:pfd}, the following corollary to Hilbert's weak Nullstellensatz can be used to eliminate factors from the denominators if the $q_i$ have no common zero.

\begin{lem}[Nullstellensatz certificate]\label{lem:NSS_cert}
	Polynomials $f_1,\dots,f_m\in K[\mathbf x]$ have no common zero in $\overline K^d$ if and only if the generated  ideal is trivial, i.e. $\left<f_1,\dots,f_m\right>=\left<1\right>=K[\mathbf x]$.\\
	In this case there exist polynomials $h_1,\dots,h_m\in K[\mathbf x]$ such that
	\begin{equation}1=\sum_{i=1}^mh_if_i\end{equation}
	Call $(h_1,\dots,h_m)$ a \textbf{Nullstellensatz certificate} for $(f_1,\dots,f_m)$.
\end{lem}
\begin{proof}
	This is exactly the weak Nullstellensatz \cite[\S4.1~Theorem 1]{CLO} with the exception, that we require $h_i\in K[\mathbf x]$ instead of $h_i\in\overline K[\mathbf x]$. However the equation $1=\sum_{i=1}^mh_if_i$ can be seen as a set of linear equations (with coefficients in $K$) in the coefficients of the polynomials $h_i$ and by the weak Nullstellensatz we know, that it is solvable over $\overline K$. But then it is solvable over $K$ as well, since all the coefficients in these linear equations lie in $K$. Hence we may assume $h_1,\dots,h_m\in K[\mathbf x]$.
\end{proof}

Given a rational function $f/g$ as in Theorem \ref{thm:pfd} for which the irreducible factors $q_i$ of $g$ have no common zero in $\overline K^d$, we know that $q_1^{e_1},\dots,q_m^{e_m}$ have no common zero as well and if $(h_1,\dots,h_m)$ is a Nullstellensatz certificate, we can simply multiply $f$ by $1=\sum_{k=1}^mh_kq_k^{e_k}$ to get a decomposition
\begin{equation}\label{eq:NSSdecomp}
\frac f g=\frac{f\cdot\sum_{k=1}^mh_kq_k^{e_k}}{\prod_{i=1}^mq_i^{e_i}}=\sum_{k=1}^m\frac{f\cdot h_k}{\prod_{i=1,i\not=k}^mq_i^{e_i}}
\end{equation}
where each denominator contains only $m-1$ different irreducible factors.

To calculate this decomposition (Algorithm \ref{algo:NSSdecompStep}), we compute a reduced Gr\"obner basis $G$ of $\left<q_1^{e_1},\dots,q_m^{e_m}\right>$ as well as the transformation matrix $T$ from the original ideal generators $q_1^{e_1},\dots,q_m^{e_m}$ to $G$. This can be done with Buchberger's algorithm for the computation of Gr\"obner bases as implemented in the \textsc{Singular} function \texttt{liftstd}.

\begin{algorithm}
	\caption{\texttt{NSSdecompStep} (Nullstellensatz decomposition step)}
	\label{algo:NSSdecompStep}
	\begin{flushleft}
		\textbf{Input:} rational function $f/g$ where $f,g\in K[\mathbf x]$ and $g=\prod_{i=1}^m q_i^{e_i}$ for irreducible $q_i\in K[\mathbf x]$
		
		\textbf{Output:} set of rational functions with sum $f/g$
	\end{flushleft}
	\begin{algorithmic}[1]
		\STATE {calculate the reduced Gr\"obner basis $G$ of $\left<q_1^{e_1}, \dots, q_m^{e_m}\right>$ as well as the transformation matrix $T$ from the generators $q_i^{e_i}$ to $G$}
		\IF {$G=\{c\}$ for $c\in K$ (so $\left<q_1^{e_1}, \dots, q_m^{e_m}\right>=\left<1\right>$)}
		\STATE {using $T$, find polynomials $h_i$ such that $\sum_{i=1}^m h_i q_i^{e_i}=1$ (namely $h_i=T_{i1}/c$)}
		\RETURN {$\left\{\frac{f\cdot h_k}{\prod_{i=1,i\not=k}^mq_i^{e_i}}\middle| k=1, \dots, m\right\}$}
		\ELSE
		\RETURN {$\left\{\frac f g\right\}$}
		\ENDIF
	\end{algorithmic}
\end{algorithm}

Note that $\left<q_1^{e_1},\dots,q_m^{e_m}\right>=\left<1\right>$ if and only if $G=\{c\}$ for $c\in K$ (constant polynomial). This follows directly from the uniqueness of reduced Gr\"obner bases. 

Repeated application of Algorithm \ref{algo:NSSdecompStep} will yield a decomposition satisfying condition $(1)$ in Theorem \ref{thm:pfd}. For $(2)$ let's assume $f/g$ is a rational function and $g=\prod_{i=1}^mq_i^{e_i}$ as in Theorem \ref{thm:pfd}. There is a simple criterion to test for algebraic dependence:
\begin{lem}[Jacobian criterion]\label{lem:jac_crit}
	A set of $m$ polynomials $\{f_1,\dots,f_m\}\subseteq K[\mathbf x]$ is algebraically independent if and only if the Jacobian matrix $\left(\frac{\partial f_i}{\partial x_j}\right)_{i,j}\in K[\mathbf x]^{m\times d}$ has rank $m$ over the field $K(\mathbf x)$ of rational functions. (A proof can be found in \cite{MR1215329}.)
\end{lem}

\begin{cor}\label{cor:alg_dep}
	A set of polynomials $\{q_1,\dots,q_m\}\subset K[\mathbf x]$ is algebraically dependent if and only if $\{q_1^{e_1},\dots,q_m^{e_m}\}$ is $(e_i\in\mathbb N_{>0})$.
\end{cor}
\begin{proof}
	This follows directly from Lemma \ref{lem:jac_crit} since 
	\begin{equation}\left(\frac{\partial \left(q_i^{e_i}\right)}{\partial x_j}\right)_{i,j}=\left(e_iq_i^{e_i-1}\frac{\partial q_i}{\partial x_j}\right)_{i,j}=T\cdot\left(\frac{\partial q_i}{\partial x_j}\right)_{i,j}\end{equation}
	where $T$ is the invertible diagonal matrix $\textup{diag}(e_1q_1^{e_1-1},\dots,e_mq_m^{e_m-1})\in K(\mathbf x)^{m\times m}$. So the Jacobian matrices of $\{q_1,\dots,q_m\}$ and $\{q_1^{e_1},\dots,q_m^{e_m}\}$ have the same rank over $K(\mathbf x)$.
\end{proof}

If now the factors $q_1,\dots,q_m$ of the denominator $g$ are algebraically dependent, then so are $q_1^{e_1},\dots,q_m^{e_m}$ and if $p\in K[\mathbf y]=K[y_1,\dots,y_m]$ is an annihilating polynomial we can write 
\begin{align}\label{eq:alpha}
p=c_\alpha\mathbf y^\alpha+\sum_{\substack{\beta\in\mathbb N^m\\ \deg(p)\ge|\beta|\ge|\alpha|}} c_\beta\mathbf y^\beta\qquad (c_\alpha,c_\beta\in K, c_\alpha\not=0)
\end{align}
such that $c_\alpha\mathbf y^\alpha$ is one of the terms of smallest degree (using multi-indices $\beta\in\mathbb N^m$, so $\deg(\mathbf y^\beta)=|\beta|=\beta_1+\dots+\beta_m$). Writing $\mathbf q$ for the vector $(q_1^{e_1},\dots,q_m^{e_m})$, it holds
\begin{align}\label{eq:algDependDecomp}
0 = p(\mathbf q)\quad &\Leftrightarrow\quad c_\alpha\mathbf q^\alpha =-\sum_{\beta} c_\beta\mathbf q^\beta\nonumber\\
&\Leftrightarrow\quad 1=-\sum_{\beta} \frac{c_\beta\mathbf q^\beta}{c_\alpha\mathbf q^\alpha}=-\sum_{\beta} \frac{c_\beta}{c_\alpha}\prod_{i=1}^m\frac{q_i^{e_i\beta_i}}{q_i^{e_i\alpha_i}}\nonumber\\
&\Rightarrow\quad \frac f g=-\sum_{\beta} \frac{c_\beta}{c_\alpha}f\prod_{i=1}^m\frac{q_i^{e_i\beta_i}}{q_i^{e_i(\alpha_i+1)}}
\end{align}
Since $\mathbf y^\alpha$ has minimal degree, for every $\beta$ occurring in Equation (\ref{eq:algDependDecomp}) it holds $\beta_i\ge\alpha_i+1$ for at least one index $i$. Therefore the factor $q_i$ does not appear in the denominator of the corresponding term. So we obtain a sum of rational functions with at most $m-1$ different irreducible factors in their denominators and thus, as with Algorithm \ref{algo:NSSdecompStep}, repeated application of this step leads to a decomposition satisfying condition $(2)$ in Theorem \ref{thm:pfd}.
But in order to turn this into an algorithm, we need a way of computing annihilating polynomials:

\begin{lem}[annihilating polynomials]\label{lem:annpoly}
	The annihilating polynomials of a fixed tuple $(f_1,\dots,f_m)$ of polynomials in $K[\mathbf x]=K[x_1,\dots,x_d]$ are precisely the elements of the ideal \[\left<y_1-f_1,\dots,y_m-f_m\right>_{K[\mathbf x,\mathbf y]}\cap K[\mathbf y]\quad\text{where}\quad \mathbf y=(y_1,\dots,y_m).\]
\end{lem}
\begin{proof}
	Let $p\in K[\mathbf y]$ be an annihilating polynomial for $f_1,\dots,f_m$ and write $\mathbf f=(f_1,\dots,f_m)$. Define $\tilde p(\mathbf x,\mathbf y)=p(\mathbf f-\mathbf y)\in K[\mathbf x,\mathbf y]$. Then $\tilde p(\mathbf x,\mathbf 0)=p(\mathbf f)=0$ and thus every term of $\tilde p$ must be divisible by some $y_i$ ($1\le i\le m$), hence $\tilde p\in\left<y_1,\dots,y_m\right>_{K[\mathbf x,\mathbf y]}$. But since $\tilde p(\mathbf x,\mathbf f-\mathbf y)=p(\mathbf y)=p$ (replacing each $y_i$ by $f_i-y_i$), we get $p\in\left<f_1-y_1,\dots,f_m-y_m\right>_{K[\mathbf x,\mathbf y]}$.
	
	Now assume $p\in \left<y_i-f_1,\dots,y_m-f_m\right>_{K[\mathbf x,\mathbf y]}\cap K[\mathbf y]$. Then $p=\sum_{i=1}^ma_i\cdot(y_i-f_i)$ for some $a_i\in K[\mathbf x,\mathbf y]$ and thus $p(\mathbf f)=\sum_{i=1}^ma_i(\mathbf x,\mathbf f)\cdot(f_i-f_i)=0$.
\end{proof}

\begin{lem}[elimination ordering]\label{lem:elimord}
	Call a monomial ordering for the polynomial ring $K[\mathbf x,\mathbf y]=K[x_1,\dots,x_d,y_1,\dots,y_m]$ an \textbf{elimination ordering} for the variables $x_1,\dots,x_d$ if $L(p)\in K[\mathbf y]$ implies already $p\in K[\mathbf y]$ for any polynomial $p\in K[\mathbf x,\mathbf y]$. 
	
	If now $G$ is a (reduced) Gr\"obner basis of an ideal $I\subseteq K[\mathbf x,\mathbf y]$ with respect to such an ordering, then $G\cap K[\mathbf y]$ is a (reduced) Gr\"obner basis of the ideal $I\cap K[\mathbf y]\subseteq K[\mathbf y]$.
	
	(A proof can be found in \cite[Lemma~1.8.3]{KPLUS-KLU01-000633484}.)
\end{lem} 

\begin{algorithm}
	\caption{\texttt{algDependDecompStep} (algebraic dependence decomposition step)}
	\label{algo:algDependDecompStep}
	\begin{flushleft}
		\textbf{Input:} rational function $f/g$ where $f,g\in K[\mathbf x]$ and $g=\prod_{i=1}^m q_i^{e_i}$ for irreducible $q_i\in K[\mathbf x]$
		
		\textbf{Output:} set of rational functions with sum $f/g$
	\end{flushleft}
	\begin{algorithmic}[1]
		\IF {$\textup{rank}\left(\frac{\partial q_i}{\partial x_j}\right)_{i\le m,j\le d}<m$}
		\STATE {calculate the reduced Gr\"obner basis $G$ of $\left<y_1-q_1^{e_1}, \dots, y_m-q_m^{e_m}\right>\subseteq K[\mathbf x,\mathbf y]$\\ with respect to an elimination ordering for $x_1,\dots, x_d$ ($\mathbf y=(y_1,\dots,y_m)$)}
		\STATE {$G'=G\cap K[\mathbf y]$}
		\STATE {$p=$ some element of $G'$ (choose a ``simple'' one, e.g. with smallest degree)}
		\STATE {write $p=c_{\alpha}y^{\alpha}+\sum_{\beta}c_{\beta}y^{\beta}$ where $y^\alpha$ has minimal degree}
		\RETURN {$\left\{-\frac{c_\beta}{c_\alpha}f\prod_{i=1,i\not=k}^m q_i^{e_i\cdot(\beta_i-\alpha_i-1)}\middle|\beta\in\mathbb{N}^m\right\}$}
		\ELSE
		\RETURN {$\left\{\frac f g\right\}$}
		\ENDIF
	\end{algorithmic}
\end{algorithm}

In order to calculate the rank of the Jacobian matrix in line $1$ of Algorithm \ref{algo:algDependDecompStep}, we can test, if the syzygy module of the $K[\mathbf x]$-module generated by the rows of the Jacobian matrix (that is the module of all $K[\mathbf x]$-linear relations of the rows) is zero (e.g. with the \textsc{Singular} command \texttt{syz}). Instead of calculating the rank of the Jacobian, we could also just check whether $G'$ is empty, however the derivatives $\frac{\partial q_i}{\partial x_j}$ are in general of much lower degree than $q_i^{e_i}$, so using the Jacobian criterion is cheaper, especially for small factors $q_i$. Also, if $d<m$ the criterion becomes trivial since the rank is at most $d$.

The previous two strategies to decompose a rational function only decrease the size of the denominators while leaving the numerator mostly untouched. To simplify the numerators as well, it makes sense to do a (reduced) division with remainder of the numerator $f$ by a Gr\"obner basis $G$ of the ideal $\left<q_1,\dots,q_m\right>$ generate by all the irreducible factors in the denominator. This gives a division expression \begin{equation}f=r+\sum_{g\in G}a_gg\qquad(r,a_g\in K[\mathbf x])\end{equation}
as in Definition \ref{defn:divrem_GB}. Rewriting this in terms of the ideal generators $q_1,\dots,q_m$ we get 
\begin{align}\label{eq:numeratorDecomp}
f&=r+\sum_{k=1}^mb_kq_k\qquad(b_k\in K[\mathbf x])\nonumber\\
\Rightarrow\quad\frac f g &=\frac{r}{\prod_{i=1}^m q_i^{e_i}} + \sum_{k=1}^m \frac{b_k}{q_k^{(e_k-1)}\prod_{i=1,i\not=k}^mq_i^{e_i}}
\end{align}
The first term with numerator $r$ already fulfills condition $(3)$ of Theorem \ref{thm:pfd}  and all other terms have one irreducible factor $q_i$ less in their denominator. Thus repeated application of Algorithm \ref{algo:numeratorDecompStep} results in a decomposition satisfying $(3)$.

\begin{algorithm}
	\caption{\texttt{numeratorDecompStep} (numerator decomposition step)}
	\label{algo:numeratorDecompStep}
	\begin{flushleft}
		\textbf{Input:} rational function $f/g$ where $f,g\in K[\mathbf x]$ and $g=\prod_{i=1}^m q_i^{e_i}$ for irreducible $q_i\in K[\mathbf x]$
		
		\textbf{Output:} set of rational functions with sum $f/g$
	\end{flushleft}
	\begin{algorithmic}[1]
		\STATE {calculate the reduced Gr\"obner basis $G$ of $\left<q_1, \dots, q_m\right>$ as well as the transformation matrix $T$ from the generators $q_i$ to $G$}
		\STATE {divide $f$ by $G$ (reduced division with remainder) to get a division expression\\ $f=r+\sum_{g\in G}a_gg$ where $r,a_g\in K[x_1,\dots,x_d]$ and $r$ is reduced w.r.t. $G$}
		\STATE {using $T$ and $(a_g)_{g\in G}$, find polynomials $b_k$ such that $f=r+\sum_{k=1}^mb_kq_k$}
		\STATE {replace $(b_1,\dots,b_m)\in K[\mathbf x]^m$ by its remainder after division by a Gr\"obner basis of the syzygy module of $(q_1,\dots,q_m)$}
		\RETURN {$\left\{\frac{b_k}{q_k^{(e_k-1)}\prod_{i=1,i\not=k}^mq_i^{e_i}}\middle| k=1, \dots, m, \;b_k\not=0\right\}\cup\left\{\frac{r}{g}\right\}$}
	\end{algorithmic}
\end{algorithm}

\begin{rem}[syzygy reduction]\label{rem:syzygy_reduction}
	In order to further simplify the coefficients $b_k$, in line $4$ of Algorithm \ref{algo:numeratorDecompStep} we replace $(b_1,\dots,b_m)\in K[\mathbf x]^m$ by its remainder after division by (a Gr\"obner basis of) the syzygy module 
	$\left\{(s_1,\dots,s_m)\middle|\sum_{i=1}^ms_iq_i=0\right\}\subset K[\mathbf x]^m$
	of $(q_1,\dots,q_m)$
	(extending the notions of monomial orderings, Gr\"obner bases and division with remainder from $K[\mathbf x]$ to the $K[\mathbf x]$-module $K[\mathbf x]^m$ as described in \cite[Chapter 2.3]{KPLUS-KLU01-000633484}). After doing this, the polynomials $b_k$ will still satisfy Equation (\ref{eq:numeratorDecomp}), since we just changed $(b_1,\dots,b_m)$ by an element $(s_1,\dots,s_m)$ of the syzygy module and $\sum_{i=1}^ms_iq_i=0$. This step is optional, but in practice we found, that reducing the coefficients by the syzygy module can dramatically reduce the runtime of Algorithm \ref{algo:pfd}. In \textsc{Singular} we can just apply the procedures \texttt{syz} and \texttt{std} to calculate a Gr\"obner basis of the syzygy module and reduce the coefficients $b_k$ with \texttt{reduce}.
\end{rem}

Using all three decomposition techniques one after the other yields Algorithm \ref{algo:pfd} which calculates a partial fraction decomposition fulfilling conditions $(1)$, $(2)$ and $(3)$, finally proving Theorem \ref{thm:pfd}.

However, in practice the calculation of annihilating polynomials can be quite slow if the degrees of the polynomials $q_i^{e_i}$ get too big. Therefore it is more efficient to do an additional ``short'' numerator decomposition before the algebraic dependence decomposition (see Algorithm \ref{algo:pfd}), in order to simplify the denominators. For this we repeatedly apply Algorithm \ref{algo:shortNumeratorDecompStep}, which is identical to Algorithm \ref{algo:numeratorDecompStep} with the exception, that whenever the remainder $r$ is nonzero, we return the input and do not decompose further since the term corresponding to $r$ would not have a smaller denominator anyway. Note that this is only effective, because most of the rational functions we are interested in (i.e. those arising from IBP-reductions) have the property, that the numerator is already contained in the ideal $\left<q_1,\dots,q_m\right>$, such that the remainder $r$ becomes $0$. Thus, while it is not needed in order to get a decomposition fulfilling the conditions of Theorem \ref{thm:pfd}, the insertion of a short numerator decomposition (lines $6$ to $8$ in Algorithm \ref{algo:pfd}) reduces runtimes.

\begin{algorithm}
	\caption{\texttt{shortNumeratorDecompStep} (short numerator decomposition step)}
	\label{algo:shortNumeratorDecompStep}
	\begin{flushleft}
		\textbf{Input:} rational function $f/g$ where $f,g\in K[\mathbf x]$ and $g=\prod_{i=1}^m q_i^{e_i}$ for irreducible $q_i\in K[\mathbf x]$
		
		\textbf{Output:} set of rational functions with sum $f/g$
	\end{flushleft}
	\begin{algorithmic}[1]
		\STATE {calculate the reduced Gr\"obner basis $G$ of $\left<q_1, \dots, q_m\right>$ as well as the transformation matrix $T$ from the generators $q_i$ to $G$}
		\STATE {divide $f$ by $G$ (reduced division with remainder) to get a division expression\\ $f=r+\sum_{g\in G}a_gg$ where $r,a_g\in K[x_1,\dots,x_d]$ and $r$ is reduced w.r.t. $G$}
		\IF {$r\not=0$}
		\RETURN {$\left\{\frac f g\right\}$}
		\ENDIF
		\STATE {using $T$ and $(a_g)_{g\in G}$, find polynomials $b_k$ such that $f=r+\sum_{k=1}^mb_kq_k$}
		\STATE {replace $(b_1,\dots,b_m)\in K[\mathbf x]^m$ by its remainder after division by a Gr\"obner basis of the syzygy module of $(q_1,\dots,q_m)$}
		\RETURN {$\left\{\frac{b_k}{q_k^{(e_k-1)}\prod_{i=1,i\not=k}^mq_i^{e_i}}\middle| k=1, \dots, m, \;b_k\not=0\right\}$}
	\end{algorithmic}
\end{algorithm}

\begin{algorithm}
	\caption{Partial fraction decomposition}
	\label{algo:pfd}
	\raggedright
	\begin{flushleft}
		\textbf{Input:} rational function $f/g$ where $f,g\in K[x_1,\dots,x_d]$
		
		\textbf{Output:} partial fraction decomposition as a set of rational functions.
	\end{flushleft}
	\begin{algorithmic}[1]
		\STATE {factorize $g=\prod_{i=1}^m q_i^{e_i}$ where $q_i\in K[x_1, \dots, x_d]$ are irreducible\\ (and represent all denominators in the following steps in factorized form)}
		
		\STATE {$D=\{f/g\}$}
		\WHILE {$\exists s\in D$ such that $\left|\texttt{NSSdecompStep}(s)\right|>1$} 
		\STATE {$D=\texttt{NSSdecompStep}(s)\cup D\backslash\{s\}$\\and merge elements of $D$ with equal denominators}
		\ENDWHILE
		
		\WHILE {$\exists s\in D$ such that $\left|\texttt{shortNumeratorDecompStep}(s)\right|>1$} 
		\STATE {$D=\texttt{shortNumeratorDecompStep}(s)\cup D\backslash\{s\}$\\and merge elements of $D$ with equal denominators}
		\ENDWHILE
		
		\WHILE {$\exists s\in D$ such that $\left|\texttt{algDependDecompStep}(s)\right|>1$} 
		\STATE {$D=\texttt{algDependDecompStep}(s)\cup D\backslash\{s\}$\\and merge elements of $D$ with equal denominators}
		\ENDWHILE
		
		\WHILE {$\exists s\in D$ such that $\left|\texttt{numeratorDecompStep}(s)\right|>1$} 
		\STATE {$D=\texttt{numeratorDecompStep}(s)\cup D\backslash\{s\}$\\and merge elements of $D$ with equal denominators}
		\ENDWHILE
		
		\RETURN {$D$}
	\end{algorithmic}
\end{algorithm}

\begin{lem}\label{lem:termination}
	Algorithm \ref{algo:pfd} terminates for any input $f/g$ and returns a partial fraction decomposition of $f/g$ satisfying all three conditions in Theorem \ref{thm:pfd}.
\end{lem}
\begin{proof}
	As shown above, Algorithms $1$ to $4$ applied to any rational function $f/g$ always return a set of rational functions with sum $f/g$. So in Algorithm \ref{algo:pfd}, at all times the elements of $D$ sum up to the input of the algorithm.
	
	It is also easy to see, that if Algorithm \ref{thm:pfd} terminates, the returned decomposition $D$ indeed fulfils conditions $(1)$ to $(3)$: If no term of the decomposition is decomposed further when applying Algorithm \ref{algo:NSSdecompStep} (i.e. the first \texttt{while} loop terminates), then by Lemma \ref{lem:NSS_cert} in each denominator the irreducible factors $q_i$ have a common zero. Similarly, if a rational function is not decomposable by Algorithm \ref{algo:algDependDecompStep}, then by Lemmata \ref{lem:annpoly} and \ref{lem:elimord} and Corollary \ref{cor:alg_dep} the $q_i$ are algebraically independent. And finally, if a rational function $f/g$ is not decomposed further in Algorithm \ref{algo:numeratorDecompStep}, this means that $f=r$ in Equation (\ref{eq:numeratorDecomp}) and thus the numerator is reduced with respect to the ideal generated by the factors $q_i$ in the denominator.
	
	Also, Algorithm \ref{algo:algDependDecompStep}, \ref{algo:numeratorDecompStep} and \ref{algo:shortNumeratorDecompStep} only ever change the denominators by removing irreducible factors (and changing their exponents) and thus preserve the properties $(1)$ and $(2)$ in Theorem \ref{thm:pfd}. (If $q_1,\dots,q_m$ have a common zero or are algebraically independent, then so are $q_1,\dots,q_{m-1}$.)
	
	It remains to show that all \texttt{while} loops terminate. As argued above, each term in the decomposition returned by Algorithm \ref{algo:NSSdecompStep} has fewer different irreducible factors in its denominator than the input and thus after applying \texttt{NSSdecompStep} to each element of $D$, all terms have at most $m-1$ different irreducible factors $q_i$ and thus the first loop terminates by induction on $m$. The same argument also works for \texttt{algDependDecompStep}. In \texttt{numeratorDecompStep} each element of the returned decomposition has one less irreducible factor in its denominator than the input with the exception of the term corresponding to the remainder $r$ in Equation (\ref{eq:numeratorDecomp}), which has the same denominator as the input. However terms of this form are not decomposed further (since in Equation (\ref{eq:numeratorDecomp}) $r$ is already reduced with respect to $\left<q_1,\dots,q_m\right>$) and can thus be disregarded in the argument. Now by induction on $\sum_{i=1}^me_i$ the fourth loop terminates as well. A very similar argument works for \texttt{shortNumeratorDecompStep}.
\end{proof}

Since the calculation of annihilating polynomials can still be quite slow for some of the more complicated rational functions, we make the following modification to the algorithm.

\begin{rem}[simplified \texttt{algDependDecompStep}] \label{rem:simplified_algDep}
	In Algorithm \ref{algo:algDependDecompStep} it is also possible to use an annihilating polynomial for $q_1,\dots,q_m$ rather than $q_1^{e_1},\dots,q_m^{e_m}$. Instead of Equation~\ref{eq:algDependDecomp} we then get the decomposition 
	\begin{equation}\label{eq:simple_algDependDecomp}
	\frac f g=-\sum_{\beta} \frac{c_\beta}{c_\alpha}f\prod_{i=1}^m\frac{q_i^{\beta_i}}{q_i^{\alpha_i+e_i}}.
	\end{equation}
	where $p=c_\alpha\mathbf y^\alpha+\sum_\beta c_\beta \mathbf y^\beta$ is the annihilating polynomial and $c_\alpha y^\alpha$ a term of minimal degree  as in Equation (\ref{eq:alpha}).
	Since the polynomials $q_i$ are of lower degree than $q_i^{e_i}$, this will speed up the calculation of annihilating polynomials at the cost of needing more steps in the algebraic dependence decomposition in Algorithm \ref{algo:pfd}, since the number of different irreducible denominator factors then does not decrease in every step. (If $\beta_i<\alpha_i+e_i$ for all $i$, it stays the same.) In fact it is not at all clear that Algorithm \ref{algo:pfd} terminates with the simplified \texttt{algDependDecompStep} and indeed this depends on the choice of $\alpha$.
	However, if $\alpha$ is chosen minimal with respect to the graded reverse lexicographic ordering $>_{grevlex}$ on $K[\mathbf y]$ as defined in Equation (\ref{eq:grevlex_defn}), it can be shown, that Algorithm \ref{algo:pfd} still terminates with a correct decomposition:
\end{rem}

\begin{proof} 
	All we have to show is that the third \texttt{while} loop in Algorithm \ref{algo:pfd} terminates, the rest follows as in the proof of Lemma \ref{lem:termination}. For this, take any sequence $f_1/g_1,f_2/g_2,\dots$ of rational functions such that $f_{i+1}/g_{i+1}$ is one of the terms in $\texttt{algDependDecompStep}(f_i/g_i)$. It is enough to show, that in each such sequence eventually a rational function is reached, which has fewer \textit{different} factors in its denominator or satisfies $(2)$ already. Assume this is not the case. Then all the denominators $g_i$ have the same irreducible factors (with different exponents). Thus in each call of the simplified \texttt{algDependDecompStep} the same annihilating polynomial $p$ is chosen (assuming a deterministic implementation of the algorithm). Write 
	\begin{equation}
	p=c_\alpha\mathbf y^\alpha+\sum_{j=1}^{r+s}c_{\beta^{(j)}}\mathbf y^{\beta^{(j)}} \qquad(r,s\in\mathbb N)
	\end{equation}
	where $|\beta^{(j)}|=|\alpha|$ for $j\le r$ and $|\beta^{(j)}|>|\alpha|$ for $j>r$. Since $\mathbf y^\alpha$ is minimal with respect to $>_{grevlex}$, for each $j=1,\dots,r$ there exists an index $k_j\in\{1,\dots,m\}$ such that 
	\begin{equation}\label{eq:alpha_beta}
	\alpha_{k_j}>\beta^{(j)}_{k_j} \text{ and } \alpha_i=\beta^{(j)}_i \quad\text{ for all } i>k_j.
	\end{equation}
	Without loss of generality we may assume that $k_1\ge\dots\ge k_r$. 
	If we factorize the denominators as $g_i=\prod_{l=1}^mq_l^{e_l^{(i)}}$ ($e_l^{(i)}\in\mathbb N$), then it holds 
	\begin{equation}\label{eq:dec_step}
	\frac{f_{i+1}}{g_{i+1}}=- \frac{c_{\beta^{(j_i)}}}{c_\alpha}f_i\prod_{l=1}^m\frac{q_l^{\beta_l^{(j_i)}}}{q_l^{\alpha_l+e^{(i)}_l}}\quad
	\end{equation}
	for some index $j_i$ depending on $i$. Since we assumed that in $g_{i+1}$ no irreducible factor vanishes, it holds
	\begin{equation}\label{eq:exponents_formula}
	e^{(i+1)}_l=\alpha_l+e^{(i)}_l-\beta^{(j_i)}_l>0\qquad(l=1,\dots,m)
	\end{equation}
	If $j_i>r$, then the sum of all exponents decreases: $\sum_{l=1}^me_l^{(i+1)}<\sum_{l=1}^me_l^{(i)}$. If $j_i\le r$, it stays the same. Thus the case $j_i>r$ can only occur for finitely many $i$ and after that it always holds $j_i\le r$, but then $|\beta^{(j_i)}|=|\alpha|$ and by Equations (\ref{eq:alpha_beta}) and (\ref{eq:exponents_formula}) it holds 
	\begin{equation}
	\sum_{l=1}^{k_{j_i}-1}e_l^{(i+1)}<\sum_{l=1}^{k_{j_i}-1}e_l^{(i)}\quad\text{and}\quad\sum_{l=1}^{k}e_l^{(i+1)}=\sum_{l=1}^{k}e_l^{(i)}\quad\text{for any }k\ge k_{j_i}
	\end{equation}
	Since the exponents $e_l^{(i)}$ have to stay positive and $k_1\ge\dots\ge k_r$, in the sequence $\beta^{(j_1)},\beta^{(j_2)},\dots$ the multi-index $\beta^{(1)}$ can only appear finitely often and after that $\beta^{(2)}$ can only appear finitely often and so on. But $f_1/g_1,f_2/g_2,\ldots$ was an infinite sequence, a contradiction.
	
	Thus in the sequence $f_1/g_1,f_2/g_2,\dots$ the number of different irreducible factors in the denominators decreases until a rational function is reached, that satisfies condition $(2)$ and the third \texttt{while} loop terminates after finitely many iterations.
\end{proof}

\begin{rem} The multivariate partial fractioning can be combined with the rational reconstruction scheme for commutative algebra developed in \cite{BDFP15,BDFLP16} as long as a consistent factorization patterns can be guaranteed.
\end{rem}

\begin{exmp}
	Consider the rational function
	\begin{equation}\frac{-2s_{12}s_{23}\varepsilon_5+3s_{23}s_{34}\varepsilon_5+s_{15}s_{45}\varepsilon_5+s_{23}s_{45}\varepsilon_5-s_{34}s_{45}\varepsilon_5}{8s_{12}s_{23}(-s_{15}+s_{23}+s_{34})(s_{12}-s_{45})s_{45}}\end{equation}
	This is the $(24,7)$-th entry of the IBP-matrix for the double pentagon (see Section \ref{sec:example}).\footnote{$\varepsilon_5\equiv 4 i\epsilon_{\mu\nu\rho\sigma} k_1^{\mu}k_2^{\nu}k_3^{\rho}k_4^{\sigma}$ for the five-point kinematics. This is a square root of a polynomial in Mandelstam variables $s_{ij}$'s. Here we can treat it as an irrelevant overall factor.} For better readability, replace $s_{12},s_{15},s_{23},s_{34},s_{45}$ by $x_1,x_2,x_3,x_4,x_5$ and divide by $\varepsilon_5/8$ to get the function
	\begin{equation}\frac f g = \frac{-2x_1x_3+3x_3x_4+x_2x_5+x_3x_5-x_4x_5}{x_1x_3(-x_2+x_3+x_4)(x_1-x_5)x_5}\in\mathbb R(\mathbf x)=\mathbb R(x_1,x_2,x_3,x_4,x_5)\end{equation}
	Now apply Algorithm \ref{algo:pfd} using the graded reverse lexicographic ordering with $x_1>\dots>x_5$ and employ the modification to \texttt{algDependDecompStep} discussed in Remark \ref{rem:simplified_algDep}. For simplicity we omit the syzygy reduction step described in Remark \ref{rem:syzygy_reduction}. 
	
	\begin{description}
		\item[1. factorization of the denominator] The denominator factors as $g = q_1\cdot q_2\cdot q_3\cdot q_4\cdot q_5$ where $q_1=x_1$, $q_2=x_3$, $q_3=-x_2+x_3+x_4$, $q_4=x_1-x_5$, $q_5=x_5$.
		
		\item[2. Nullstellensatz decomposition] As with most entries of this IBP-matrix, the denominators have already a common zero, namely $0$, since none of the factors $q_i$ have a constant term. Thus Algorithm \ref{algo:NSSdecompStep} (\texttt{NSSdecompStep}) does nothing.
		
		\item[3. short numerator decomposition]   $ $\linebreak
		\underline{\textit{reduced G.B. of  $\left<q_1, q_2, q_3, q_4, q_5\right>$}:}
		$G=(x_5,x_3,x_2-x_4,x_1)=(q_5,q_2,q_2-q_3,q_4+q_5)$
		
		\underline{\textit{reduced division with remainder}:}
		$f=(x_2+x_3-x_4)\cdot x_5+(-2x_1+3x_4)\cdot x_3$\\ $\left.\right.\qquad\qquad\qquad\qquad\qquad\qquad\qquad\;\;\:\,=(x_2+x_3-x_4)\cdot q_5+(-2x_1+3x_4)\cdot q_2$
		
		\begin{equation}\Rightarrow\quad\frac f g=\frac{x_2+x_3-x_4}{q_1\cdot q_2\cdot q_3\cdot q_4}+\frac{-2x_1+3x_4}{q_1\cdot q_3\cdot q_4\cdot q_5}=\frac{f_1}{g_1}+\frac{f_2}{g_2}\end{equation}

		\underline{\textit{reduced G.B. of  $\left<q_1, q_2, q_3, q_4\right>$}:}
		$G=(x_5,x_3,x_2-x_4,x_1)=(q_1-q_4,q_2,q_2-q_3,q_1)$
		
		\underline{\textit{reduced division with remainder}:}
		$f_1=1\cdot x_3+1\cdot (x_2-x_4)=2\cdot q_2+(-1)\cdot q_3$
		
		\begin{equation}\Rightarrow\quad\frac{f_1}{g_1}=\frac{2}{q_1\cdot q_3\cdot q_4}+\frac{-1}{q_1\cdot q_2\cdot q_4}\end{equation}
		
		\underline{\textit{reduced G.B. of  $\left<q_1, q_3, q_4, q_5\right>$}:}
		$G=(x_5,x_2-x_3-x_4,x_1)=(q_5,-q_3,q_4+q_5)$
		
		\underline{\textit{reduced division with remainder}:}
		$f_2=(-2)\cdot x_1+3x_4=(-2)\cdot q_4+(-2)\cdot q_5+3x_4$
		
		The remainder $3x_4$ is nonzero, so we do not decompose $f_2/g_2$ further and get in total
		\begin{equation}\frac f g = \frac{2}{q_1\cdot q_3\cdot q_4}+\frac{-1}{q_1\cdot q_2\cdot q_4}+\frac{-2x_1+3x_4}{q_1\cdot q_3\cdot q_4\cdot q_5}\end{equation}
		
		\item[4. (simplified) algebraic dependence decomposition] For the first two terms the Jacobian matrices have full rank and therefore Algorithm \ref{algo:algDependDecompStep} (\texttt{algDependDecompStep}) does nothing. For the third term the Jacobian matrix is
		\begin{equation}\left(\frac{\partial q_i}{\partial x_j}\right)_{\begin{subarray}{l} i=1,3,4,5\\j=1,2,3,4,5\end{subarray}}=\begin{pmatrix}1&0&0&0&0\\ 0&-1&1&1&0\\ 1&0&0&0&-1\\ 0&0&0&0&1\end{pmatrix}\end{equation}
		and has only rank $3$, so $q_1,q_3,q_4,q_5$ are algebraically dependent. Indeed it is obvious, that $q_1=x_1$, $q_4=x_1-x_5$ and $q_5=x_5$ are even linearly dependent and thus a possible annihilating polynomial for $q_1,q_3,q_4,q_5$ is 
		\begin{equation}p = y_1-y_3-y_4\in\mathbb R[y_1,y_2,y_3,y_4]\end{equation}
		leading to the relation 
		\begin{equation}1=\frac{q_1-q_4}{q_5}\quad\Rightarrow\quad\frac{-2x_1+3x_4}{q_1\cdot q_3\cdot q_4\cdot q_5}=\frac{-2x_1+3x_4}{ q_3\cdot q_4\cdot q_5^2}+\frac{2x_1-3x_4}{q_1\cdot q_3\cdot q_5^2}\end{equation}
		Now $\{q_3,q_4,q_5\}$ as well as $\{q_1,q_4,q_5\}$ are algebraically independent and we are done. Note that the exponent of $q_5$ increased and the number of irreducible factors in the denominators stayed the same ($4$), but the number of \emph{different} irreducible factors decreased from $4$ to $3$. Overall, we now have
		\begin{equation}\frac f g = \frac{2}{q_1\cdot q_3\cdot q_4}+\frac{-1}{q_1\cdot q_2\cdot q_4}+\frac{-2x_1+3x_4}{ q_3\cdot q_4\cdot q_5^2}+\frac{2x_1-3x_4}{q_1\cdot q_3\cdot q_5^2}\end{equation}
		and all terms fulfil conditions $(1)$ and $(2)$ of Theorem \ref{thm:pfd}.
		
		\item[5. numerator decomposition] 
		The first two numerators ($2$ and $-1$) are obviously reduced with respect to $\left<q_1,q_3,q_4\right>$ and $\left<q_1,q_2,q_4\right>$ respectively. So Algorithm \ref{algo:numeratorDecompStep} \linebreak (\texttt{numeratorDecompStep}) does nothing. For the third and fourth term we get:
		
		\underline{\textit{reduced Gr\"obner basis of  $\left<q_3, q_4, q_5\right>$}:}
		$G=(x_5,x_2-x_3-x_4,x_1)=(q_5,-q_3,q_4+q_5)$
		
		\underline{\textit{reduced division with remainder}:}
		$-2x_1+3x_4=(-2)\cdot x_1+3x_4$\\
		$\left.\right.\qquad\qquad\qquad\qquad\qquad\qquad\qquad\qquad\qquad\;\;\;\;\;=(-2)\cdot q_4+(-2)\cdot q_5+3x_4$
		\begin{equation}\Rightarrow\quad\frac{-2x_1+3x_4}{ q_3\cdot q_4\cdot q_5^2}=\frac{-2}{ q_3\cdot q_5^2}+\frac{-2}{ q_3\cdot q_4\cdot q_5}+\frac{3x_4}{ q_3\cdot q_4\cdot q_5^2}\end{equation}

		\underline{\textit{reduced Gr\"obner basis of  $\left<q_1, q_3, q_5\right>$}:}
		$G=(x_5,x_2-x_3-x_4,x_1)=(q_5,-q_3,q_1)$
		
		\underline{\textit{reduced division with remainder}:}
		$2x_1-3x_4=2\cdot x_1-3x_4=2\cdot q_1-3x_4$
		
		\begin{equation}\Rightarrow\quad\frac{2x_1-3x_4}{ q_1\cdot q_3\cdot q_5^2}=\frac{2}{ q_3\cdot q_5^2}+\frac{-3x_4}{ q_1\cdot q_3\cdot q_5^2}\end{equation}
		
		Thus, merging terms with the same denominator, we get in total
		\begin{equation}\frac f g = \frac{2}{q_1\cdot q_3\cdot q_4}+\frac{-1}{q_1\cdot q_2\cdot q_4}+\frac{-2}{ q_3\cdot q_4\cdot q_5}+\frac{3x_4}{ q_3\cdot q_4\cdot q_5^2}+\frac{-3x_4}{ q_1\cdot q_3\cdot q_5^2}\end{equation}
		where all terms satisfy conditions $(1)$ to $(3)$.
		
		\item[Example for syzygy reduction:] If we had done the syzygy reduction step, for the first step in the short numerator decomposition above, i.e. the decomposition
		\begin{equation}\frac f g=\frac{x_2+x_3-x_4}{q_1\cdot q_2\cdot q_3\cdot q_4}+\frac{-2x_1+3x_4}{q_1\cdot q_3\cdot q_4\cdot q_5}\end{equation}
		arising from the division expression
		\begin{equation}\label{eq:relation}
		f=(-2x_1+3x_4)\cdot q_2+(x_2+x_3-x_4)\cdot q_5,
		\end{equation}
		we would have to calculate a Gr\"obner basis of the syzygy module of the ideal $\left<q_1, q_2, q_3, q_4, q_5\right>$. Using \textsc{Singular} we get the reduced Gr\"obner basis
		\begin{dmath}
			\left\{\left(-1, 0, 0, 1, 1\right),
			\left(-x_3, x_5, 0, x_3, 0\right),
			\left(0, x_2-x_3-x_4, x_3, 0, 0\right),\\
			\left(-x_2+x_4, x_5, -x_5, x_2-x_4, 0\right),
			\left(-x_3, x_1, 0, 0, 0\right),\\
			\left(x_2-x_3-x_4, 0, x_1, 0, 0\right),
			\left(-x_1+x_5, 0, 0, x_1, 0\right)\right\}
		\end{dmath}
		with respect to the graded reverse lexicographic ordering with priority to monomials (see \cite[Definition~2.3.1]{KPLUS-KLU01-000633484}). The original relation (\ref{eq:relation}) corresponds to the module element $\left(0, -2x_1+3x_4, 0, 0, x_2+x_3-x_4\right)$ and dividing by the Gr\"obner basis yields the remainder $\left(-2x_3, 3x_4+2x_5, -x_5, 0, 0\right)$. So we would use the relation 
		\begin{equation}
		f=(-2x_3)\cdot q_1+(3x_4+2x_5)\cdot q_2+(-x_5)\cdot q_3
		\end{equation}
		instead of relation (\ref{eq:relation}), leading to the decomposition step
		\begin{equation}
		\frac f g=\frac{-2x_3}{q_2\cdot q_3\cdot q_4\cdot q_5}+\frac{3x_4+2x_5}{q_1\cdot q_3\cdot q_4\cdot q_5}+\frac{-x_5}{q_1\cdot q_2\cdot q_4\cdot q_5}.
		\end{equation}
		This simple example may not show, why syzygy reduction should be an advantage, since here the decomposition seems to get longer, but, as mentioned above, for more complicated input functions we observe a significant improvement of the performance of the algorithm.
	\end{description}
\end{exmp}

\section{Examples}
\label{sec:example}

In the following, we discuss the application of the partial fraction
decomposition to IBP reduction coefficients of Feynman integrals,
in examples of various complexity, and with and without use of a UT basis.


\subsection{A baby example: Nonplanar two-loop four-point 
with an external massive leg}
In this subsection, we present a ``baby'' example, one-massive crossed box, showing how partial fraction decomposition simplifies the IBP reduction coefficients. 

\begin{figure}[H]
\centering
\includegraphics[width=0.5\textwidth]{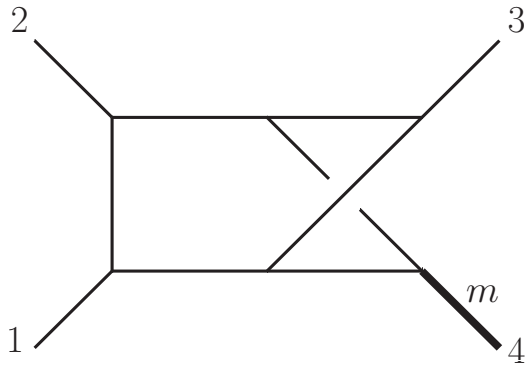}
\caption{One-massive crossed box.}
\end{figure}

The physical kinematic conditions are that $p_1,p_2$ and $p_3$ are massless, while $p_4^2=m^2$, $2p_1\cdot p_2=s$ and $2p_2\cdot p_3=t$. The propagators are
\begin{equation}
\begin{aligned}
&D_1=k_1^2,\,\,
D_2=(k_1-p_1)^2,\,\,
D_3=(k_1-p_1-p_2)^2,\,\,
D_4=k_2^2,\\
&D_5=(k_2+p_1+p_2+p_3)^2,\,\,
D_6=(k_2+k_1+p_3)^2,\,\,
D_7=(k_1+k_2)^2,\\
&D_8=(k_1-p_1-p_2-p_3)^2,\,\,
D_9=(k_2+p_1)^2\\
\end{aligned}
\end{equation}
The parameters are thus $\epsilon=(4-d)/2$, $s$, $t$, and $m^2$.

We study the IBP reduction coefficients of integrals in the sector
$(1,1,1,1,1,1,1,0,0)$ with the ISP degrees up to $5$. This is a simple
example, the IBP reduction can be easily done with LiteRed/FIRE6
\cite{Lee:2013mka, Smirnov:2019qkx}. There are $29$ master integrals. With LiteRed/FIRE6's master
integral choice, the byte size of the IBP reduction coefficients is
around 9.5MB.

We discuss the coefficients in more detail, listing the irreducible
denominator factors (poles) below:
\begin{equation}
\begin{aligned}
&\epsilon +1,\quad-2 \epsilon -1,\quad-2 \epsilon ,\quad1-2 \epsilon ,\quad2-2 \epsilon ,\quad3-2 \epsilon ,\\
&2-\epsilon ,\quad-4 \epsilon -1,\quad1-4 \epsilon ,\quad3-4 \epsilon
,\quad 1-3 \epsilon ,\quad 2-3 \epsilon ,\\
&m^2,\quad m^2-s,\quad s,\quad -10 m^2 \epsilon -6 m^2+12 s \epsilon +8 s,\\
&m^2-t,\quad m^2-s-t,\quad t,\quad s+t,\quad m^2 s-m^2t-s^2-s t,
\end{aligned}
\end{equation}
It is not surprising that there is a pole $ -10 m^2 \epsilon -6 m^2+12
s \epsilon +8 s$ , with the dependence in both
$\epsilon$ and the kinematic parameters. There is also a nonlinear
pole, $m^2 s-m^2t-s^2-s t$ occurring in the list above.

We then convert the IBP reduction coefficients to a UT basis. It is
easy to find the UT basis via leading singularity analysis or Wasser's
dlog algorithm \cite{Wasser:2018qvj,Henn:2020lye}. The IBP reduction
coefficients of the UT basis clearly have simpler poles:
\begin{equation}
\begin{aligned}
&\epsilon -2,\quad \epsilon -1,\quad 2 \epsilon -3,\quad 2 \epsilon -1,\quad 3 \epsilon -2,\quad 3 \epsilon -1,\quad 4\epsilon -3,\quad 4 \epsilon -1,\\
&m^2,\quad m^2-s,\quad s,\quad m^2-t,\quad m^2-s-t,\quad t,\quad s+t
\end{aligned}
\end{equation}
We find the previously occurring factor $ -10 m^2 \epsilon -6 m^2+12
s \epsilon +8 s$ with mixed dependence in $\epsilon$ and kinematic
variables is now absent. Furthermore, all the kinematic dependent poles
are symbol letters, as seen by a comparison with the canonical
differential equation. The previously occurring denominator factor $m^2
s-m^2t-s^2-s t$, which is {\it not} a symbol letter, is also absent.

Note that in this example, the size of the IBP coefficients with
respect to the UT
basis, is around 9.0MB. By converting to the UT basis, the byte size
of coefficients does not decrease much, but the denominator structure
becomes much simpler.


We then apply our implementation of our partial fractioning algorithm to
the IBP reduction coefficients, both with respect to the Laporta and the UT
basis, to simplify the coefficients. 
\begin{itemize}
\item After applying the algorithm, the size of IBP coefficients with
respect to the Laporta basis (LiteRed/FIRE6) is shortened from $9.5$MB
to $3.0$MB ($2.7$MB if indexed), simplified by about a factor of $3.4$.
\item Converting to the UT basis and then applying the algorithm, the resulting
coefficients are only of size $1.9$ MB ($1.5$
MB if indexed). With
respect to the original Laporta basis a $6.5$-times size reduction.
\end{itemize}
This example indicates that our method works for both the Laporta basis and
UT basis, but the size reduction ratio is larger for UT
basis. Since this is a baby example, our method runs fast in both cases.

\subsection{A cutting-edge example: Nonplanar two-loop five-point }

In this section, we present a computationally cutting-edge example, the two-loop five-point nonplanar double pentagon. The diagram is shown in Figure \ref{double pentagon}.
\begin{figure}[H]
\label{double pentagon}
\centering
\includegraphics[width=0.5\textwidth]{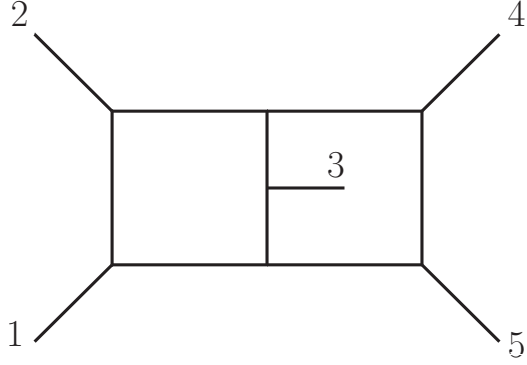}
\caption{2-loop 5-point nonplanar double pentagon}
\end{figure}
All external and internal lines are massless. The kinematic conditions are $2p_1\cdot p_2=s_{12}$, $2p_2\cdot p_3=s_{23}$, $2p_3\cdot p_4=s_{34}$, $2p_4\cdot p_5=s_{45}$ and $2p_1\cdot p_5=s_{15}$. The propagators are
\begin{equation}
\begin{aligned}
&D_1=l_1^2,\quad
D_2=(l_1-p_1)^2,\quad
D_3=(l_1-p_{12})^2,\quad
D_4=l_2^2,\\
&D_5=(l_2-p_{123})^2,\quad
D_6=(l_2-p_{1234})^2,\quad
D_7=(l_1-l_2)^2,\\
&D_8=(l_1-l_2+p_3)^2,\quad
D_9=(l_1-p_{1234})^2,\quad
D_{10}=(l_2-p_1)^2,\\
&D_{11}=(l_2-p_{12})^2.
\end{aligned}
\end{equation}
Where $p_{i\dots j}=\sum_{k=i}^j p_k$. A UT basis for this diagram and its symbol form was found in ref.~\cite{Abreu:2018aqd,
Chicherin:2018old}, and the analytic expressions for the master
integrals were obtained in ref.~\cite{Chicherin:2018old}.

Relying on the module intersection IBP reduction method and its 
implementation in the \textsc{Singular}-\textsc{GPI-Space} framework for massively parallel computations, the analytic IBP reduction coefficients were calculated for
the integrals with ISP up to the degree $4$ in the sector
$(1,1,1,1,1,1,1,1,0,0,0)$ in ref.~\cite{Bendle:2019csk}.\footnote{We
note that during the preparation of this manuscript, the degree-$5$
analytic IBP reduction of the same diagram, was calculated in a recent paper
\cite{Klappert:2020nbg} by the reduction of a new type of IBP system
in the block triangular form \cite{Guan:2019bcx}.} The size of
the IBP reduction coefficients with respect to a Laporta basis is
$2.4$GB (with all parameters analytic).

When reducing target integrals to the Laporta basis, we found some
``mixed'' denominator factors in the coefficients, which are mixtures
of the spacetime parameter $\epsilon$ and kinematic variables. They are listed in the following:
\begin{equation}
\begin{aligned}
&3s_{12} \epsilon -s_{23}\epsilon -s_{12}+s_{23},\\
&3s_{12}\epsilon -s_{23} \epsilon -4s_{34} \epsilon -3 s_{45} \epsilon+s_{23}+s_{34},\\
&3s_{12} \epsilon+4 s_{23} \epsilon +s_{34}\epsilon -3 s_{45} \epsilon-s_{23}-s_{34},\\
&3s_{12} \epsilon+8 s_{23} \epsilon -3 s_{45}\epsilon -s_{12}-4 s_{23}+s_{45},\\
&4s_{12} \epsilon +s_{23} \epsilon-s_{45} \epsilon -2s_{12}-s_{23}+s_{45}
\end{aligned}
\end{equation}

As we have observed in ref.~\cite{Bendle:2019csk}, if we reduce the target
integrals to the UT basis, the size of coefficients is reduced to
$712$MB. More importantly, the IBP reduction coefficients with respect
to UT basis have no ``mixed'' poles and all kinematic denominators are
symbol letters. The irreducible factors in the IBP reduction
coefficients with respect to the UT basis are given as,
\begin{equation}\label{dpUTpole}
\begin{aligned}
&\epsilon -1,\quad
2 \epsilon -1,\quad
3\epsilon -1,\quad
4 \epsilon -1,\quad
4\epsilon+1,\quad
s_{12},\quad
s_{15},\quad
s_{15}-s_{23},\\
&s_{23},\quad
s_{12}+s_{23},\quad
s_{12}-s_{34},\quad
s_{12}+s_{15}-s_{34},\quad
s_{15}-s_{23}-s_{34},\quad
s_{34},\\
&s_{23}+s_{34},\quad
s_{12}-s_{45},\quad
s_{23}-s_{45},\quad
s_{12}+s_{23}-s_{45},\quad
s_{12}-s_{15}+s_{23}-s_{45},\\
&s_{12}-s_{34}-s_{45},\quad
s_{12}+s_{15}-s_{34}-s_{45},\quad
s_{45},\quad
s_{15}-s_{23}+s_{45},\quad
s_{34}+s_{45},\\
&s_{15}^2s_{12}^2+s_{23}^2 s_{12}^2-2s_{15} s_{23} s_{12}^2-2 s_{23}^2s_{34} s_{12}+2 s_{15} s_{23}s_{34} s_{12}-2 s_{15}^2 s_{45}s_{12}\\&+2 s_{15} s_{23} s_{45}s_{12}+2 s_{15} s_{34} s_{45}s_{12}+2 s_{23} s_{34} s_{45}s_{12}+s_{23}^2 s_{34}^2+s_{15}^2s_{45}^2+s_{34}^2 s_{45}^2\\&-2s_{15} s_{34} s_{45}^2-2 s_{23}s_{34}^2 s_{45}+2 s_{15} s_{23}s_{34} s_{45}
\end{aligned}
\end{equation}
We see that except for the factors only in $\epsilon$, all 
other factors are (powers of) even symbol letters. (The symbol letters of all
two-loop five-point massless topology were obtained in
ref.~\cite{Chicherin:2017dob}.) Note that the last factor above is the
Gram determinant $G(1,2,3,4)$. Since
\begin{equation}
\label{eq:2}
G(1,2,3,4) = \epsilon_5^2 \equiv-16 (\epsilon_{\mu\nu\rho\sigma} k_1^{\mu}k_2^{\nu}k_3^{\rho}k_4^{\sigma})^2
\end{equation}
and $\epsilon_5$ is a symbol letter, the last factor $G(1,2,3,4)$ is a
power of a symbol letter. In addition to the computation
in ref.~\cite{Bendle:2019csk}, we further checked the IBP reduction of
the integrals with ISP up to the degree $5$ for the same diagram, and
this pole structure property still holds. 


At this point, it is interesting to compare the pole structure in the
UT basis with
the same double pentagon diagram reduced in the basis choice of
ref.~\cite{Usovitsch:2020jrk}. In ref.~\cite{Usovitsch:2020jrk} there
are $7$ nonlinear irreducible factors in the IBP reduction denominators, such as
$-1+s_{34}+s_{45}+s_{34} s_{45}$, $-s_{23}-s_{45}-s_{23}
s_{45}+s_{45}^2$, $s_{15}-s_{23}+s_{23} s_{34}-s_{15} s_{45}+s_{34}
s_{45}$, $\ldots$ ($s_{12}$ is set to be 1). Except the Gram
determinant $G(1,2,3,4)$, the other $6$ nonlinear factors in ref.~\cite{Usovitsch:2020jrk} are not symbol letters. 

Despite the fact that the size of the coefficients is already simplified by about
$3$ times, if we change the basis from Laporta to UT, the 
coefficients are still huge with a size of $712$MB. We now apply our improved 
Leinartas' algorithm to shorten these coefficients with respect to the
UT basis. The size of the coefficients is magically shortened to only
$24$MB ($19$MB in indexed form). Compared with the $2.4$GB IBP reduction file we started out with, those IBP
reduction coefficients are made simpler by over 100 times
\footnote{The resulting indexed partial fraction result can be
downloaded from \url{https://www.dropbox.com/s/wku5o20g0vaggtl/xb_deg4_UT_pfd.zip}.}!

As a comparison, without using the UT basis, our algorithm can also
reduce the IBP coefficients size from $2.4$G to $864$MB. However, the
reduction ratio is not as dramatic and the running time is much
longer. 

In Appendix \ref{An explicit example of the size reduction} we present a visual impression about how powerful our algorithm is: a 5-page-long coefficient is shortened to only 9 lines.

\subsection{An elliptic example: two-loop four-point with a top quark loop and a
pair of external massive legs}
Our algorithm also works well for cases without the existence of a UT
basis. In this subsection, we present an elliptic example, the
double box diagram with one massive internal loop and two massive external lines. The diagram is
\begin{figure}[H]
\centering
\includegraphics[width=0.5\textwidth]{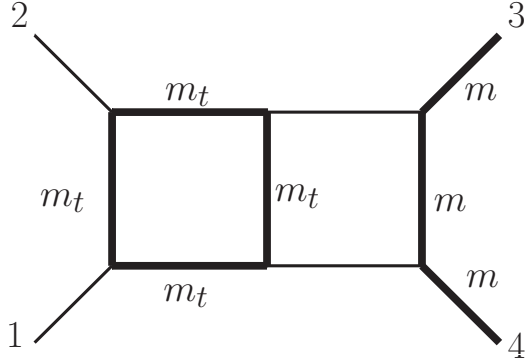}
\caption{elliptic double box}
\end{figure}
The kinetic conditions are $p_1^2=p_2^2=0$, $p_3^2=m^2$, $2p_1\cdot p_2=s$, $2p_2\cdot p_3=t-m^2$ and $2p_2\cdot p_3=m^2-s-t$. The propagators are
\begin{equation}
\begin{aligned}
&D_1=(l_1+p_2)^2-m_t^2,\quad
D_2=l_1^2-m_t^2,\quad
D_3=(l_1+p_1+p_2)^2-m_t^2,\\
&D_4=(l_1+l_2)^2-m_t^2,\quad
D_5=l_2^2,\quad
D_6=(l_2-p_1-p_2)^2,\quad
D_7=(l_2+p_3)^2-m^2,\\
&D_8=(l_1+p_2-p_3)^2-m^2,\quad
D_9=(l_2-p_2+p_3)^2
\end{aligned}
\end{equation}
The parameters are $\epsilon, s, t, m^2, m_t^2$. It is clear that there are fully massive sunset sub-diagrams in this
topology and the UT basis does not exist.

We have reduced integrals in the sector $(1,1,1,1,1,1,1,0,0)$, with the
ISP degree up to $5$ to the Laporta basis, using FIRE6. The size of the resulting coefficients is in total 175MB. 

In applying our algorithm to shorten these coefficients, it is
important to pull the nonlinear factors out and do the partial
fractions over the linear factors. After applying our algorithm, the
size of simplified coefficients is reduced to
only $24$ MB. This is also a significant simplification, by about $7$
times in byte-size.

This example indicates that although one should prefer a UT basis in doing
partial fraction, for diagrams without the existence of UT basis,
this algorithm is still powerful.


\subsection{Performance of the algorithm}

In this section, we summarize the computing resources used for our
examples, and the reduction ratio in different formats.

In all examples a \textsc{Singular} implementation of Algorithm \ref{algo:pfd} with the improvements described in Remarks \ref{rem:syzygy_reduction} and \ref{rem:simplified_algDep} was used. Table \ref{tab:resources} shows the resources used for applying the algorithm to all matrix entries one after the other or in parallel using 32 cores. Due to the simple form of parallelism, the computation will scale similarly up to the number of entries.

\begin{table}[h]
	\centering
	\caption{Runtime and memory used in Algorithm \ref{algo:pfd}}
	\label{tab:resources}
	\begin{tabular}{|l|l|l|l|l|l|}
		\hline example & \multicolumn{2}{|c|}{sequential} & \multicolumn{2}{|c|}{parallel (32 cores)} & memory  for just\\
		& runtime & memory &runtime & memory & reading input file
		\\\hline\hline
		$\diamond$ xbox1m (Laporta)  & 929 s & 67 MB & 74 s & 171 MB & 38 MB \\\hline
		xbox1m (UT) & 1073 s & 53 MB & 48 s & 151 MB & 30 MB \\\hline\hline
		$\diamond$ dpentagon (Laporta) & -- & -- & 137 h & 44.3 GB & 9.63 GB \\\hline
		dpentagon (UT) & 1672 min & 5.39 GB & 99 min & 10.24 GB & 2.45 GB \\\hline\hline
		$\diamond$ dbox elliptic (Laporta) & 304 min & 1.24 GB & 41 min & 2.18 GB & 495 MB \\\hline
	\end{tabular}
\end{table}

When comparing the time taken for each decomposition step, we found that the short numerator decomposition (Algorithm \ref{algo:shortNumeratorDecompStep}) needs 60-95 \% of the total runtime. 

Especially for large numerators and small (by degree) denominator
factors, partial fraction decomposition can drastically reduce the
size of IBP-matrices, as can be seen in Table
\ref{tab:file_size}. Since most of the irreducible factors in the
denominators are linear, it makes sense to leave any nonlinear
factors untouched in the algorithm. In the tables in this subsection, the symbol
``$\diamond$'' means that we leave the nonlinear factors untouched in our partial
fraction algorithm. The phrase ``(Laporta)'' or ``(UT)'' means that we are dealing with
coefficients in a Laporta integral basis or a UT basis, respectively.

The use of a UT basis typically leads to a shorter runtime (Table \ref{tab:resources}) and also reduces the size of the output (Table \ref{tab:file_size}). Finally, instead of writing out the denominator, we can just store in the data structure the indices $i$ and exponents $e_i$ of the irreducible factors $q_i$ appearing in each denominator together with all factors $q_i$, which also reduces the size (last column in Table \ref{tab:file_size}).

\begin{table}[h]
	\caption{size comparison of input/output of Algorithm \ref{algo:pfd}}
	\centering
	\label{tab:file_size}
	\begin{tabular}{|l|l|l|l|}
		\hline example & input & output & output (indexed)
		\\\hline\hline
		$\diamond$ xbox1m (Laporta) & 9.51 MB & 2.91 MB (30.7 \%) & 2.75 MB (29.0 \%) \\\hline
		xbox1m (UT) & 9.04 MB & 1.80 MB (19.9 \%) & 1.53 MB (16.9 \%) \\\hline\hline
		$\diamond$ double pentagon (Laporta)  & 2.42 GB & 864 MB (35.7 \%) & 851 MB (35.2 \%) \\\hline
		double pentagon (UT) & 712 MB & 28.0 MB (3.93 \%) & 19.8 MB (2.78 \%)  \\\hline\hline
		
		
		%

		$\diamond$ dbox elliptic (Laporta)&  175 MB & 24.1 MB (13.8 \%) & 23.3 MB (13.3 \%) \\\hline

	\end{tabular}
\end{table}

We also find an interesting phenomenon that zipping both the input and
output files in some examples leads to a further increase of the relative size reduction (see Table \ref{tab:file_size_zipped}).

\begin{table}[h]
	\caption{size comparison of input/output of Algorithm \ref{algo:pfd} after zipping}
	\centering
	\label{tab:file_size_zipped}
	\begin{tabular}{|l|l|l|l|}
		\hline example & input (zipped) & output (zipped) &
		output
		(indexed,
		zipped)
		\\\hline\hline
		$\diamond$  xbox1m (Laporta) & 2.90 MB & 829 KB (28.6 \%) & 815 KB (28.1 \%) \\\hline
		xbox1m (UT) & 2.43 MB & 369 KB (15.2 \%) & 351 KB (14.5 \%) \\\hline\hline
		$\diamond$ double pentagon (Laporta)  & 648 MB & 216 MB (33.3 \%) & 215 MB (33.2 \%) \\\hline
		double pentagon (UT) & 212 MB & 4.42 MB (2.09 \%) & 3.99 MB (1.89 \%) \\\hline\hline
		$\diamond$ dbox elliptic (Laporta)  & 50.5 MB & 7.60 MB (15.0 \%) & 7.51 MB (14.9 \%) \\\hline
		
	\end{tabular}
\end{table}

\section{Summary and Discussion}
\label{sec:summary}
In this manuscript, we develop an improved 
Leinartas' algorithm of multivariate partial fraction and present an
modern implement of this algorithm, to simplify the
complicated analytic IBP reduction coefficients in multi-loop
computations. We show that for cases with the existence or without the existence
of the UT basis, our algorithm works well to reduce the IBP reduction
coefficients size.

We observe that in the cases we studied, the IBP reduction
coefficients in the UT basis have simple structures: (1) the spacetime
dimension parameter $D$ factorizes out in the denominators (2) the rest
irreducible factors in the denominators are a subset of the symbol
letters. Thus usually the UT basis provides a simpler denominator
factor list, and our algorithm works particularly well with shorter
running time and higher reduction ratio. In complicated examples, our
algorithm achieves dramatic size reduction in the coefficients of
IBPs.

We expect that our algorithm will have broad applications in
the multi-loop IBP computations, to get easier-to-use analytic
reduction results
and make the numeric evaluation much faster. 

We present a {\sc Singular} library for our algorithm of multivariate partial
fraction. It can be used for simplifying IBP coefficients in general
purposes. Furthermore, we expect that the partial-fraction library can be used to
simplify multiloop integrand and the transcendental function
coefficients in scattering amplitudes, as the partial-fraction examples
shown in \cite{Abreu:2019odu}. We expect that this library can be combined with current
finite field and rational reconstruction packages ~\cite{vonManteuffel:2014ixa,Klappert:2019emp,Klappert:2020aqs,Peraro:2019svx,
  Smirnov:2019qkx} for multiloop
scattering amplitude computations. Our library would also find
applications in analytic computations outside scattering amplitudes, in broader research areas
in theoretical physics. 

Besides this partial fraction library, in the future we
will also develop an arithmetic  library to perform arithmetic
computations of ration functions in partial fraction form and keep
the output in partial fraction form.

It would also be interesting to study the IBP reduction coefficients in a
 UT basis in details. After partial fraction, it seems that each team
 in a coefficient looks much simpler. It is then of theoretical interests
 to relate these terms to the leading singularities of Feynman
 integrals.

\acknowledgments
We acknowledge Johannes Henn for enlightening discussions and his valuable
comments on our original manuscript. We are also grateful to
David Kosower, Kasper Larsen, Roman Lee, Hui Luo, Yanqing Ma, Andreas
von Manteuffel,  Erik Panzer, 
Robert Schabinger, Alexander Smirnov, Vladimir Smirnov, Johann Usovitsch and
Pascal Wasser for very helpful discussions on related topics.

The work of JB was
supported by Project II.5 of SFB-TRR 195 Symbolic Tools in Mathematics and their Application"
of the German Research Foundation (DFG). The work of YZ was support from the NSF of China through Grant No. 11947301.

\appendix
\section{Manual of the Partial Fractioning {\sc Singular} Library}

In this section, we give a short outline of how to use the features of the \textsc{Singular} library \texttt{pfd.lib}. Together with a complete documentation, it can be downloaded from
\begin{quotation}
  \centerline{\scalebox{0.9}{\url{https://raw.githubusercontent.com/Singular/Singular/spielwiese/Singular/LIB/pfd.lib}}}
\end{quotation}
and should be placed within the user's \textsc{Singular} search path. The latest release of \textsc{Singular} can be downloaded from the \textsc{Singular} website \url{https://www.singular.uni-kl.de}.\footnote{A separate download will not be necessary in the future, since the library will be part of the next release of \textsc{Singular} (version $4.1.4$).} The website also provides an online documentation of \textsc{Singular} and all libraries distributed with the release. 

After starting up \textsc{Singular}, the library can be loaded by typing \texttt{LIB "pfd.lib";} at the \textsc{Singular} promt. The main algorithm for partial fraction decomposition can be accessed via the procedure \texttt{pfd}. This procedure takes as input two polynomials (numerator and denominator) and returns a partial fraction decomposition encoded as a list. The first entry is a list containing the denominator factors, the second entry is a list of summands, each of which is encoded as a list of numerator, indices of denominator factors and exponents of denominator factors. The decomposition can be displayed using the procedure \texttt{displaypfd}, and checked with the procedure \texttt{checkpfd}, which verifies whether a rational function (first argument) is mathematically equal to a decomposition returned by \texttt{pfd} (second argument) and returns a boolean (see Example \ref{exmp:pfd}). An example of how to use a procedure can be displayed by typing \texttt{example <name-of-procedure>;} at the \textsc{Singular} prompt. 
\begin{exmp}[\texttt{pfd}]\label{exmp:pfd}
	In the following example, we calculate a partial fraction decomposition of the rational function $(x^2+3xy-y^2) / ((x+1)\cdot (2x+y))\in\mathbb Q(x,y)$ with respect to the graded reverse lexicographic ordering:
	
	\begin{algorithmic}
		\STATE \texttt{> LIB "pfd.lib";}
		\STATE \texttt{> ring r = 0,(x,y),dp;}
		\STATE \texttt{> poly f = x^2+3xy-y^2;}
		\STATE \texttt{> poly g = (x+1)*(2x+y);}
		\STATE \texttt{> list d = pfd(f,g);}
		\STATE \texttt{> displaypfd(d);}
		\STATE \texttt{~ (1/2)}
		\STATE \texttt{+ (9/2y+9) / (q2)}
		\STATE \texttt{+ (-y-5) / (q1)}
		\STATE \texttt{+ (-9) / (q1*q2)}
		\STATE \texttt{where}
		\STATE \texttt{q1 = x+1}
		\STATE \texttt{q2 = 2x+y}
		\STATE \texttt{ }
		\STATE \texttt{> checkpfd(list(f,g),d);}
		\STATE \texttt{1}
	\end{algorithmic}
	In the first line we create the underlying polynomial ring by specifying the characteristic of the coefficient field (\texttt{0} for $\mathbb Q$), the names of the variables and a monomial ordering (\texttt{dp} for the graded reverse lexicographic ordering). Note that the result produced by the algorithm depends on the ordering.
	
	The second argument (the denominator polynomial) can alternatively be given in factorized form (as a list of a \textsc{Singular} \texttt{ideal} generated by irreducible non-constant polynomials and an \texttt{intvec} containing the exponents), in case the denominator factors are known to the user. As an example
	\begin{algorithmic}
		\STATE \texttt{> pfd(x+2*y, (x+y)^2*(x-y)^3);}
	\end{algorithmic}
	is equivalent to
	\begin{algorithmic}
		\STATE \texttt{> pfd(x+2*y, list(ideal(x+y,x-y), intvec(2,3)));}
	\end{algorithmic}
\end{exmp}

Using the procedure \texttt{pfdMat}, we can calculate the decompositions of a matrix of rational functions. The computation is done in parallel, relying on the library \texttt{parallel.lib}.\footnote{A version relying on parallelism implemented via the Singular/GPI-Space framework \cite{singgpi} is under development. This version will allow the use of our algorithm on HPC clusters.} By default, \texttt{pfdMat} also calls \texttt{checkpfd} for each decomposition and ignores nonlinear denominator factors (as described in Section \ref{sec:example}). 

The input of \texttt{pfdMat} is the name of a \texttt{.txt}-file (as a \textsc{Singular} \texttt{string}), which contains the matrix as a list of lists (row by row) enclosed in the symbols ``\texttt{\{}'' and ``\texttt{\}}'', and separated by commas (see Example \ref{exmp:pfdMat}). Each rational function has to be an expression of the form ``\texttt{a}'', ``\texttt{(a)/(b)}'', ``\texttt{(b)^(-n)}'' or ``\texttt{(a)*(b)^(-n)}'',
where ``\texttt{n}'' stands for a positive integer and ``\texttt{a}'', ``\texttt{b}'' stand for arbitrary polynomials (using the operators ``\texttt{+}'', ``\texttt{-}'', ``\texttt{*}'', ``\texttt{^}'' and brackets ``\texttt{(}'',``\texttt{)}''). A minus sign ``\texttt{-}'' followed by such an expression is also allowed.
Note that the library also has options to use the \textsc{Singular} binary serialization data format \texttt{.ssi} for highly efficient input and output from within \textsc{Singular}.

There are four optional arguments which determine
whether \texttt{checkpfd} should be applied (\texttt{-1}: exact test, \texttt{0}: do not apply \texttt{checkpfd}, positive integer: do this amount of probabilistic tests, default value is \texttt{-1}),
whether nonlinear factors should be extracted (\texttt{1} or \texttt{0}, default value \texttt{1}),
whether additional output files should be created (integer from \texttt{1} to \texttt{4}, default value \texttt{1}) and
whether the algorithm should be run in parallel over all matrix entries (\texttt{1} or \texttt{0}, default value \texttt{1}). The options should be specified in this order. The third optional argument (integer from \texttt{1} to \texttt{4}) controls the output files created:
\begin{itemize}[noitemsep,topsep=0pt]
	\item[\texttt{1:}] The output, that is, the matrix containing the decompositions, is stored in a \texttt{.txt}-file in indexed form (as described in Section \ref{sec:example}). The denominator factors are saved in a separate file and a logfile is created, which protocols runtimes and memory usage.
	\item[\texttt{2:}] Additionally, the decompositions are saved in non-indexed form.
	\item[\texttt{3:}] Additional \texttt{.ssi}-files containing the input and output matrix as well as some intermediate results are created.
	\item[\texttt{4:}] Additionally to mode \texttt{3}, for every rational function, the result of \texttt{pfd} is immediately saved in a seperate \texttt{.ssi}-file. (This creates a file for every matrix entry.)
\end{itemize}
For more details refer to the documentation of the library.

Before calling \texttt{pfdMat}, a polynomial ring must be defined (as in Example \ref{exmp:pfd}) such that the variable names match the names used in the input file. Furthermore, with the command \texttt{setcores(n);} the number of processor cores used for the parallelization can be set to an integer \texttt{n}. By default, all cores are used.

\begin{exmp}[\texttt{pfdMat}]\label{exmp:pfdMat}
	Suppose the file \texttt{test.txt} in the search path contains the string\\
	\centerline{``\texttt{\{\{(x+y)/(x^2-x*y),\;-(x^2*y+1)/y,\;x^2\},\;\{(x+y+1)/(y^2),\;0,\;(x^2*y-y^3)^(-1)\}\}}''}
	representing a $2\times 3$ matrix of rational functions.
	We then can calculate the decompositions of the $6$ matrix entries (in parallel on $4$ cores) as follows:
	\begin{algorithmic}
		\STATE \texttt{> LIB "pfd.lib";}
		\STATE \texttt{> setcores(4);}
		\STATE \texttt{> ring r = 0,(x,y),dp;}
		\STATE \texttt{> pfdMat("test.txt");}
	\end{algorithmic}
	The procedure \texttt{pfdMat} creates two output files: \texttt{test_pfd_indexed.txt} containing the matrix of partial fraction decompositions as the string
	\begin{dmath*}
		\text{``}\texttt{\{\{(2)/(q2) + (-1)/(q1), (-x^2) + (-1)/(q3), (x^2)\},}\\
		\texttt{\{(1)/(q3) + (x+1)/(q3^2), (0), (1)/(q3*q4^2) + (2)/(q2*q4^2)\}\}}\text{''}
	\end{dmath*}
	and \texttt{test_denominator_factors.txt} containing the factors \texttt{q1}, \texttt{q2}, \texttt{q3}, \texttt{q4} in the form
	\centerline{``\texttt{q1 = x; \quad q2 = x-y; \quad q3 = y; \quad q4 = x+y;}''.}
	As explained above, \texttt{pfdMat} has four optional arguments with default values \texttt{-1,1,1,1}. If the third argument is set to \texttt{2} as in ``\texttt{pfdMat("test.txt",-1,1,2,1);}'', an additional file \texttt{test_pfd.txt} is created, which contains the result in non indexed form as the string
	\begin{dmath*}
		\text{``}\texttt{\{\{(2)/((x-y)) + (-1)/((x)), (-x^2) + (-1)/((y)), (x^2)\},}\\
		\texttt{\{(1)/((y)) + (x+1)/((y)^2), (0), (1)/((y)*(x+y)^2) + (2)/((x-y)*(x+y)^2)\}\}}\text{''}.
	\end{dmath*}
\end{exmp}

\begin{rem}
	In an IBP-matrix, typically the total number of distinct irreducible factors occurring in the denominators of all the rational functions is fairly small (e.g. for the double pentagon example in UT-basis, there are only $25$ different irreducible denominator factors occurring in the $26\times108=2808$ matrix entries). The procedure \texttt{pfdMat} uses this fact to speed up the factorization of the denominators by dividing by factors obtained from factorizing ``simpler'' polynomials in order to factorize the larger polynomials.
\end{rem}

\section{An Explicit Example of the Size Reduction}
\label{An explicit example of the size reduction}

In this Appendix, we explicitly show an IBP reduction coefficient, before and after our partial fraction computations to see the size reduction.

This example is from the $2$-loop $5$-point nonplanar double pentagon IBP reduction (Figure.~\ref{double pentagon}). 
We choose a UT basis $I_i$, $i=1,\ldots, 108$. The integral to be reduced is $I[1,1,1,1,1,1,1,1,-2,-1,0]$,
\begin{equation}
   I[1,1,1,1,1,1,1,1,-2,-1,0] = \sum_{i=1}^{108} c_i I_i\, .
 \label{A1}
\end{equation}
Because of the UT basis, the coefficients shown in~\eqref{A1} are simpler than the coefficients in FIRE basis. However, the coefficients are still very complicated. We pick up the particular coefficient $c_{107}$ as an example. Before the partial fraction, in the usually combined form, we have a huge expression of $c_{107}$. It is can listed in about five pages as:
\begin{dmath}
c_{107}=\big(-64  { \epsilon}  {s_{15}}^3  {s_{45}}^7-16  {s_{15}}^3  {s_{45}}^7+64  { \epsilon}  {s_{34}}^3  {s_{45}}^7+16  {s_{34}}^3  {s_{45}}^7-4  { \epsilon}  {s_{12}}  {s_{15}}^2  {s_{45}}^7- {s_{12}}  {s_{15}}^2  {s_{45}}^7-4  { \epsilon}  {s_{12}}  {s_{34}}^2  {s_{45}}^7- {s_{12}}  {s_{34}}^2  {s_{45}}^7-192  { \epsilon}  {s_{15}}  {s_{34}}^2  {s_{45}}^7-48  {s_{15}}  {s_{34}}^2  {s_{45}}^7-48  { \epsilon}  {s_{23}}  {s_{34}}^2  {s_{45}}^7-12  {s_{23}}  {s_{34}}^2  {s_{45}}^7-48  { \epsilon}  {s_{15}}^2  {s_{23}}  {s_{45}}^7-12  {s_{15}}^2  {s_{23}}  {s_{45}}^7+12  { \epsilon}  {s_{12}}  {s_{15}}  {s_{23}}  {s_{45}}^7+3  {s_{12}}  {s_{15}}  {s_{23}}  {s_{45}}^7+192  { \epsilon}  {s_{15}}^2  {s_{34}}  {s_{45}}^7+48  {s_{15}}^2  {s_{34}}  {s_{45}}^7+8  { \epsilon}  {s_{12}}  {s_{15}}  {s_{34}}  {s_{45}}^7+2  {s_{12}}  {s_{15}}  {s_{34}}  {s_{45}}^7-12  { \epsilon}  {s_{12}}  {s_{23}}  {s_{34}}  {s_{45}}^7-3  {s_{12}}  {s_{23}}  {s_{34}}  {s_{45}}^7+96  { \epsilon}  {s_{15}}  {s_{23}}  {s_{34}}  {s_{45}}^7+24  {s_{15}}  {s_{23}}  {s_{34}}  {s_{45}}^7-64  { \epsilon}  {s_{15}}^4  {s_{45}}^6-16  {s_{15}}^4  {s_{45}}^6+128  { \epsilon}  {s_{34}}^4  {s_{45}}^6+32  {s_{34}}^4  {s_{45}}^6+248  { \epsilon}  {s_{12}}  {s_{15}}^3  {s_{45}}^6+62  {s_{12}}  {s_{15}}^3  {s_{45}}^6-196  { \epsilon}  {s_{12}}  {s_{34}}^3  {s_{45}}^6-49  {s_{12}}  {s_{34}}^3  {s_{45}}^6-320  { \epsilon}  {s_{15}}  {s_{34}}^3  {s_{45}}^6-80  {s_{15}}  {s_{34}}^3  {s_{45}}^6-352  { \epsilon}  {s_{23}}  {s_{34}}^3  {s_{45}}^6-76  {s_{23}}  {s_{34}}^3  {s_{45}}^6+12  { \epsilon}  {s_{12}}^2  {s_{15}}^2  {s_{45}}^6+3  {s_{12}}^2  {s_{15}}^2  {s_{45}}^6-12  { \epsilon}  {s_{12}}^2  {s_{23}}^2  {s_{45}}^6-3  {s_{12}}^2  {s_{23}}^2  {s_{45}}^6+96  { \epsilon}  {s_{15}}^2  {s_{23}}^2  {s_{45}}^6+24  {s_{15}}^2  {s_{23}}^2  {s_{45}}^6-96  { \epsilon}  {s_{12}}  {s_{15}}  {s_{23}}^2  {s_{45}}^6-24  {s_{12}}  {s_{15}}  {s_{23}}^2  {s_{45}}^6+8  { \epsilon}  {s_{12}}^2  {s_{34}}^2  {s_{45}}^6+2  {s_{12}}^2  {s_{34}}^2  {s_{45}}^6+192  { \epsilon}  {s_{15}}^2  {s_{34}}^2  {s_{45}}^6+48  {s_{15}}^2  {s_{34}}^2  {s_{45}}^6+144  { \epsilon}  {s_{23}}^2  {s_{34}}^2  {s_{45}}^6+36  {s_{23}}^2  {s_{34}}^2  {s_{45}}^6+640  { \epsilon}  {s_{12}}  {s_{15}}  {s_{34}}^2  {s_{45}}^6+160  {s_{12}}  {s_{15}}  {s_{34}}^2  {s_{45}}^6+264  { \epsilon}  {s_{12}}  {s_{23}}  {s_{34}}^2  {s_{45}}^6+54  {s_{12}}  {s_{23}}  {s_{34}}^2  {s_{45}}^6+784  { \epsilon}  {s_{15}}  {s_{23}}  {s_{34}}^2  {s_{45}}^6+172  {s_{15}}  {s_{23}}  {s_{34}}^2  {s_{45}}^6+80  { \epsilon}  {s_{15}}^3  {s_{23}}  {s_{45}}^6+20  {s_{15}}^3  {s_{23}}  {s_{45}}^6+176  { \epsilon}  {s_{12}}  {s_{15}}^2  {s_{23}}  {s_{45}}^6+44  {s_{12}}  {s_{15}}^2  {s_{23}}  {s_{45}}^6-24  { \epsilon}  {s_{12}}^2  {s_{15}}  {s_{23}}  {s_{45}}^6-6  {s_{12}}^2  {s_{15}}  {s_{23}}  {s_{45}}^6+64  { \epsilon}  {s_{15}}^3  {s_{34}}  {s_{45}}^6+16  {s_{15}}^3  {s_{34}}  {s_{45}}^6-692  { \epsilon}  {s_{12}}  {s_{15}}^2  {s_{34}}  {s_{45}}^6-173  {s_{12}}  {s_{15}}^2  {s_{34}}  {s_{45}}^6-12  { \epsilon}  {s_{12}}  {s_{23}}^2  {s_{34}}  {s_{45}}^6-3  {s_{12}}  {s_{23}}^2  {s_{34}}  {s_{45}}^6-240  { \epsilon}  {s_{15}}  {s_{23}}^2  {s_{34}}  {s_{45}}^6-60  {s_{15}}  {s_{23}}^2  {s_{34}}  {s_{45}}^6-20  { \epsilon}  {s_{12}}^2  {s_{15}}  {s_{34}}  {s_{45}}^6-5  {s_{12}}^2  {s_{15}}  {s_{34}}  {s_{45}}^6+28  { \epsilon}  {s_{12}}^2  {s_{23}}  {s_{34}}  {s_{45}}^6+7  {s_{12}}^2  {s_{23}}  {s_{34}}  {s_{45}}^6-512  { \epsilon}  {s_{15}}^2  {s_{23}}  {s_{34}}  {s_{45}}^6-116  {s_{15}}^2  {s_{23}}  {s_{34}}  {s_{45}}^6-440  { \epsilon}  {s_{12}}  {s_{15}}  {s_{23}}  {s_{34}}  {s_{45}}^6-98  {s_{12}}  {s_{15}}  {s_{23}}  {s_{34}}  {s_{45}}^6+64  { \epsilon}  {s_{34}}^5  {s_{45}}^5+16  {s_{34}}^5  {s_{45}}^5+252  { \epsilon}  {s_{12}}  {s_{15}}^4  {s_{45}}^5+63  {s_{12}}  {s_{15}}^4  {s_{45}}^5-376  { \epsilon}  {s_{12}}  {s_{34}}^4  {s_{45}}^5-94  {s_{12}}  {s_{34}}^4  {s_{45}}^5-64  { \epsilon}  {s_{15}}  {s_{34}}^4  {s_{45}}^5-16  {s_{15}}  {s_{34}}^4  {s_{45}}^5-560  { \epsilon}  {s_{23}}  {s_{34}}^4  {s_{45}}^5-116  {s_{23}}  {s_{34}}^4  {s_{45}}^5-360  { \epsilon}  {s_{12}}^2  {s_{15}}^3  {s_{45}}^5-90  {s_{12}}^2  {s_{15}}^3  {s_{45}}^5+24  { \epsilon}  {s_{12}}^2  {s_{23}}^3  {s_{45}}^5+6  {s_{12}}^2  {s_{23}}^3  {s_{45}}^5-48  { \epsilon}  {s_{15}}^2  {s_{23}}^3  {s_{45}}^5-12  {s_{15}}^2  {s_{23}}^3  {s_{45}}^5+156  { \epsilon}  {s_{12}}  {s_{15}}  {s_{23}}^3  {s_{45}}^5+39  {s_{12}}  {s_{15}}  {s_{23}}^3  {s_{45}}^5+256  { \epsilon}  {s_{12}}^2  {s_{34}}^3  {s_{45}}^5+64  {s_{12}}^2  {s_{34}}^3  {s_{45}}^5-192  { \epsilon}  {s_{15}}^2  {s_{34}}^3  {s_{45}}^5-48  {s_{15}}^2  {s_{34}}^3  {s_{45}}^5+672  { \epsilon}  {s_{23}}^2  {s_{34}}^3  {s_{45}}^5+144  {s_{23}}^2  {s_{34}}^3  {s_{45}}^5+1056  { \epsilon}  {s_{12}}  {s_{15}}  {s_{34}}^3  {s_{45}}^5+264  {s_{12}}  {s_{15}}  {s_{34}}^3  {s_{45}}^5+1164  { \epsilon}  {s_{12}}  {s_{23}}  {s_{34}}^3  {s_{45}}^5+231  {s_{12}}  {s_{23}}  {s_{34}}^3  {s_{45}}^5+1088  { \epsilon}  {s_{15}}  {s_{23}}  {s_{34}}^3  {s_{45}}^5+236  {s_{15}}  {s_{23}}  {s_{34}}^3  {s_{45}}^5-12  { \epsilon}  {s_{12}}^3  {s_{15}}^2  {s_{45}}^5-3  {s_{12}}^3  {s_{15}}^2  {s_{45}}^5+32  { \epsilon}  {s_{12}}^3  {s_{23}}^2  {s_{45}}^5+8  {s_{12}}^3  {s_{23}}^2  {s_{45}}^5-32  { \epsilon}  {s_{15}}^3  {s_{23}}^2  {s_{45}}^5-8  {s_{15}}^3  {s_{23}}^2  {s_{45}}^5-368  { \epsilon}  {s_{12}}  {s_{15}}^2  {s_{23}}^2  {s_{45}}^5-92  {s_{12}}  {s_{15}}^2  {s_{23}}^2  {s_{45}}^5+236  { \epsilon}  {s_{12}}^2  {s_{15}}  {s_{23}}^2  {s_{45}}^5+59  {s_{12}}^2  {s_{15}}  {s_{23}}^2  {s_{45}}^5-4  { \epsilon}  {s_{12}}^3  {s_{34}}^2  {s_{45}}^5- {s_{12}}^3  {s_{34}}^2  {s_{45}}^5+320  { \epsilon}  {s_{15}}^3  {s_{34}}^2  {s_{45}}^5+80  {s_{15}}^3  {s_{34}}^2  {s_{45}}^5-144  { \epsilon}  {s_{23}}^3  {s_{34}}^2  {s_{45}}^5-36  {s_{23}}^3  {s_{34}}^2  {s_{45}}^5-732  { \epsilon}  {s_{12}}  {s_{15}}^2  {s_{34}}^2  {s_{45}}^5-183  {s_{12}}  {s_{15}}^2  {s_{34}}^2  {s_{45}}^5-832  { \epsilon}  {s_{12}}  {s_{23}}^2  {s_{34}}^2  {s_{45}}^5-184  {s_{12}}  {s_{23}}^2  {s_{34}}^2  {s_{45}}^5-1152  { \epsilon}  {s_{15}}  {s_{23}}^2  {s_{34}}^2  {s_{45}}^5-240  {s_{15}}  {s_{23}}^2  {s_{34}}^2  {s_{45}}^5-884  { \epsilon}  {s_{12}}^2  {s_{15}}  {s_{34}}^2  {s_{45}}^5-221  {s_{12}}^2  {s_{15}}  {s_{34}}^2  {s_{45}}^5-464  { \epsilon}  {s_{12}}^2  {s_{23}}  {s_{34}}^2  {s_{45}}^5-80  {s_{12}}^2  {s_{23}}  {s_{34}}^2  {s_{45}}^5-416  { \epsilon}  {s_{15}}^2  {s_{23}}  {s_{34}}^2  {s_{45}}^5-104  {s_{15}}^2  {s_{23}}  {s_{34}}^2  {s_{45}}^5-2208  { \epsilon}  {s_{12}}  {s_{15}}  {s_{23}}  {s_{34}}^2  {s_{45}}^5-468  {s_{12}}  {s_{15}}  {s_{23}}  {s_{34}}^2  {s_{45}}^5+80  { \epsilon}  {s_{15}}^4  {s_{23}}  {s_{45}}^5+20  {s_{15}}^4  {s_{23}}  {s_{45}}^5-264  { \epsilon}  {s_{12}}  {s_{15}}^3  {s_{23}}  {s_{45}}^5-66  {s_{12}}  {s_{15}}^3  {s_{23}}  {s_{45}}^5-220  { \epsilon}  {s_{12}}^2  {s_{15}}^2  {s_{23}}  {s_{45}}^5-55  {s_{12}}^2  {s_{15}}^2  {s_{23}}  {s_{45}}^5+4  { \epsilon}  {s_{12}}^3  {s_{15}}  {s_{23}}  {s_{45}}^5+ {s_{12}}^3  {s_{15}}  {s_{23}}  {s_{45}}^5-128  { \epsilon}  {s_{15}}^4  {s_{34}}  {s_{45}}^5-32  {s_{15}}^4  {s_{34}}  {s_{45}}^5-200  { \epsilon}  {s_{12}}  {s_{15}}^3  {s_{34}}  {s_{45}}^5-50  {s_{12}}  {s_{15}}^3  {s_{34}}  {s_{45}}^5+60  { \epsilon}  {s_{12}}  {s_{23}}^3  {s_{34}}  {s_{45}}^5+15  {s_{12}}  {s_{23}}^3  {s_{34}}  {s_{45}}^5+192  { \epsilon}  {s_{15}}  {s_{23}}^3  {s_{34}}  {s_{45}}^5+48  {s_{15}}  {s_{23}}^3  {s_{34}}  {s_{45}}^5+988  { \epsilon}  {s_{12}}^2  {s_{15}}^2  {s_{34}}  {s_{45}}^5+247  {s_{12}}^2  {s_{15}}^2  {s_{34}}  {s_{45}}^5+128  { \epsilon}  {s_{12}}^2  {s_{23}}^2  {s_{34}}  {s_{45}}^5+32  {s_{12}}^2  {s_{23}}^2  {s_{34}}  {s_{45}}^5+512  { \epsilon}  {s_{15}}^2  {s_{23}}^2  {s_{34}}  {s_{45}}^5+104  {s_{15}}^2  {s_{23}}^2  {s_{34}}  {s_{45}}^5+796  { \epsilon}  {s_{12}}  {s_{15}}  {s_{23}}^2  {s_{34}}  {s_{45}}^5+175  {s_{12}}  {s_{15}}  {s_{23}}^2  {s_{34}}  {s_{45}}^5+16  { \epsilon}  {s_{12}}^3  {s_{15}}  {s_{34}}  {s_{45}}^5+4  {s_{12}}^3  {s_{15}}  {s_{34}}  {s_{45}}^5-20  { \epsilon}  {s_{12}}^3  {s_{23}}  {s_{34}}  {s_{45}}^5-5  {s_{12}}^3  {s_{23}}  {s_{34}}  {s_{45}}^5-192  { \epsilon}  {s_{15}}^3  {s_{23}}  {s_{34}}  {s_{45}}^5-36  {s_{15}}^3  {s_{23}}  {s_{34}}  {s_{45}}^5+1308  { \epsilon}  {s_{12}}  {s_{15}}^2  {s_{23}}  {s_{34}}  {s_{45}}^5+303  {s_{12}}  {s_{15}}^2  {s_{23}}  {s_{34}}  {s_{45}}^5+632  { \epsilon}  {s_{12}}^2  {s_{15}}  {s_{23}}  {s_{34}}  {s_{45}}^5+122  {s_{12}}^2  {s_{15}}  {s_{23}}  {s_{34}}  {s_{45}}^5-184  { \epsilon}  {s_{12}}  {s_{34}}^5  {s_{45}}^4-46  {s_{12}}  {s_{34}}^5  {s_{45}}^4+64  { \epsilon}  {s_{15}}  {s_{34}}^5  {s_{45}}^4+16  {s_{15}}  {s_{34}}^5  {s_{45}}^4-256  { \epsilon}  {s_{23}}  {s_{34}}^5  {s_{45}}^4-52  {s_{23}}  {s_{34}}^5  {s_{45}}^4-372  { \epsilon}  {s_{12}}^2  {s_{15}}^4  {s_{45}}^4-93  {s_{12}}^2  {s_{15}}^4  {s_{45}}^4-12  { \epsilon}  {s_{12}}^2  {s_{23}}^4  {s_{45}}^4-3  {s_{12}}^2  {s_{23}}^4  {s_{45}}^4-72  { \epsilon}  {s_{12}}  {s_{15}}  {s_{23}}^4  {s_{45}}^4-18  {s_{12}}  {s_{15}}  {s_{23}}^4  {s_{45}}^4+368  { \epsilon}  {s_{12}}^2  {s_{34}}^4  {s_{45}}^4+92  {s_{12}}^2  {s_{34}}^4  {s_{45}}^4-192  { \epsilon}  {s_{15}}^2  {s_{34}}^4  {s_{45}}^4-48  {s_{15}}^2  {s_{34}}^4  {s_{45}}^4+912  { \epsilon}  {s_{23}}^2  {s_{34}}^4  {s_{45}}^4+180  {s_{23}}^2  {s_{34}}^4  {s_{45}}^4+240  { \epsilon}  {s_{12}}  {s_{15}}  {s_{34}}^4  {s_{45}}^4+60  {s_{12}}  {s_{15}}  {s_{34}}^4  {s_{45}}^4+1520  { \epsilon}  {s_{12}}  {s_{23}}  {s_{34}}^4  {s_{45}}^4+308  {s_{12}}  {s_{23}}  {s_{34}}^4  {s_{45}}^4+208  { \epsilon}  {s_{15}}  {s_{23}}  {s_{34}}^4  {s_{45}}^4+52  {s_{15}}  {s_{23}}  {s_{34}}^4  {s_{45}}^4+232  { \epsilon}  {s_{12}}^3  {s_{15}}^3  {s_{45}}^4+58  {s_{12}}^3  {s_{15}}^3  {s_{45}}^4-52  { \epsilon}  {s_{12}}^3  {s_{23}}^3  {s_{45}}^4-13  {s_{12}}^3  {s_{23}}^3  {s_{45}}^4+148  { \epsilon}  {s_{12}}  {s_{15}}^2  {s_{23}}^3  {s_{45}}^4+37  {s_{12}}  {s_{15}}^2  {s_{23}}^3  {s_{45}}^4-384  { \epsilon}  {s_{12}}^2  {s_{15}}  {s_{23}}^3  {s_{45}}^4-96  {s_{12}}^2  {s_{15}}  {s_{23}}^3  {s_{45}}^4-124  { \epsilon}  {s_{12}}^3  {s_{34}}^3  {s_{45}}^4-31  {s_{12}}^3  {s_{34}}^3  {s_{45}}^4+192  { \epsilon}  {s_{15}}^3  {s_{34}}^3  {s_{45}}^4+48  {s_{15}}^3  {s_{34}}^3  {s_{45}}^4-544  { \epsilon}  {s_{23}}^3  {s_{34}}^3  {s_{45}}^4-124  {s_{23}}^3  {s_{34}}^3  {s_{45}}^4+572  { \epsilon}  {s_{12}}  {s_{15}}^2  {s_{34}}^3  {s_{45}}^4+143  {s_{12}}  {s_{15}}^2  {s_{34}}^3  {s_{45}}^4-2284  { \epsilon}  {s_{12}}  {s_{23}}^2  {s_{34}}^3  {s_{45}}^4-463  {s_{12}}  {s_{23}}^2  {s_{34}}^3  {s_{45}}^4-1392  { \epsilon}  {s_{15}}  {s_{23}}^2  {s_{34}}^3  {s_{45}}^4-264  {s_{15}}  {s_{23}}^2  {s_{34}}^3  {s_{45}}^4-1160  { \epsilon}  {s_{12}}^2  {s_{15}}  {s_{34}}^3  {s_{45}}^4-290  {s_{12}}^2  {s_{15}}  {s_{34}}^3  {s_{45}}^4-1668  { \epsilon}  {s_{12}}^2  {s_{23}}  {s_{34}}^3  {s_{45}}^4-321  {s_{12}}^2  {s_{23}}  {s_{34}}^3  {s_{45}}^4+512  { \epsilon}  {s_{15}}^2  {s_{23}}  {s_{34}}^3  {s_{45}}^4+92  {s_{15}}^2  {s_{23}}  {s_{34}}^3  {s_{45}}^4-2708  { \epsilon}  {s_{12}}  {s_{15}}  {s_{23}}  {s_{34}}^3  {s_{45}}^4-605  {s_{12}}  {s_{15}}  {s_{23}}  {s_{34}}^3  {s_{45}}^4+4  { \epsilon}  {s_{12}}^4  {s_{15}}^2  {s_{45}}^4+ {s_{12}}^4  {s_{15}}^2  {s_{45}}^4-28  { \epsilon}  {s_{12}}^4  {s_{23}}^2  {s_{45}}^4-7  {s_{12}}^4  {s_{23}}^2  {s_{45}}^4+176  { \epsilon}  {s_{12}}  {s_{15}}^3  {s_{23}}^2  {s_{45}}^4+44  {s_{12}}  {s_{15}}^3  {s_{23}}^2  {s_{45}}^4+588  { \epsilon}  {s_{12}}^2  {s_{15}}^2  {s_{23}}^2  {s_{45}}^4+147  {s_{12}}^2  {s_{15}}^2  {s_{23}}^2  {s_{45}}^4-156  { \epsilon}  {s_{12}}^3  {s_{15}}  {s_{23}}^2  {s_{45}}^4-39  {s_{12}}^3  {s_{15}}  {s_{23}}^2  {s_{45}}^4-64  { \epsilon}  {s_{15}}^4  {s_{34}}^2  {s_{45}}^4-16  {s_{15}}^4  {s_{34}}^2  {s_{45}}^4+48  { \epsilon}  {s_{23}}^4  {s_{34}}^2  {s_{45}}^4+12  {s_{23}}^4  {s_{34}}^2  {s_{45}}^4-1128  { \epsilon}  {s_{12}}  {s_{15}}^3  {s_{34}}^2  {s_{45}}^4-282  {s_{12}}  {s_{15}}^3  {s_{34}}^2  {s_{45}}^4+920  { \epsilon}  {s_{12}}  {s_{23}}^3  {s_{34}}^2  {s_{45}}^4+218  {s_{12}}  {s_{23}}^3  {s_{34}}^2  {s_{45}}^4+720  { \epsilon}  {s_{15}}  {s_{23}}^3  {s_{34}}^2  {s_{45}}^4+156  {s_{15}}  {s_{23}}^3  {s_{34}}^2  {s_{45}}^4+832  { \epsilon}  {s_{12}}^2  {s_{15}}^2  {s_{34}}^2  {s_{45}}^4+208  {s_{12}}^2  {s_{15}}^2  {s_{34}}^2  {s_{45}}^4+1508  { \epsilon}  {s_{12}}^2  {s_{23}}^2  {s_{34}}^2  {s_{45}}^4+317  {s_{12}}^2  {s_{23}}^2  {s_{34}}^2  {s_{45}}^4+416  { \epsilon}  {s_{15}}^2  {s_{23}}^2  {s_{34}}^2  {s_{45}}^4+80  {s_{15}}^2  {s_{23}}^2  {s_{34}}^2  {s_{45}}^4+2628  { \epsilon}  {s_{12}}  {s_{15}}  {s_{23}}^2  {s_{34}}^2  {s_{45}}^4+501  {s_{12}}  {s_{15}}  {s_{23}}^2  {s_{34}}^2  {s_{45}}^4+560  { \epsilon}  {s_{12}}^3  {s_{15}}  {s_{34}}^2  {s_{45}}^4+140  {s_{12}}^3  {s_{15}}  {s_{34}}^2  {s_{45}}^4+460  { \epsilon}  {s_{12}}^3  {s_{23}}  {s_{34}}^2  {s_{45}}^4+79  {s_{12}}^3  {s_{23}}  {s_{34}}^2  {s_{45}}^4-624  { \epsilon}  {s_{15}}^3  {s_{23}}  {s_{34}}^2  {s_{45}}^4-132  {s_{15}}^3  {s_{23}}  {s_{34}}^2  {s_{45}}^4+1124  { \epsilon}  {s_{12}}  {s_{15}}^2  {s_{23}}  {s_{34}}^2  {s_{45}}^4+317  {s_{12}}  {s_{15}}^2  {s_{23}}  {s_{34}}^2  {s_{45}}^4+2424  { \epsilon}  {s_{12}}^2  {s_{15}}  {s_{23}}  {s_{34}}^2  {s_{45}}^4+510  {s_{12}}^2  {s_{15}}  {s_{23}}  {s_{34}}^2  {s_{45}}^4-252  { \epsilon}  {s_{12}}  {s_{15}}^4  {s_{23}}  {s_{45}}^4-63  {s_{12}}  {s_{15}}^4  {s_{23}}  {s_{45}}^4+292  { \epsilon}  {s_{12}}^2  {s_{15}}^3  {s_{23}}  {s_{45}}^4+73  {s_{12}}^2  {s_{15}}^3  {s_{23}}  {s_{45}}^4+80  { \epsilon}  {s_{12}}^3  {s_{15}}^2  {s_{23}}  {s_{45}}^4+20  {s_{12}}^3  {s_{15}}^2  {s_{23}}  {s_{45}}^4+16  { \epsilon}  {s_{12}}^4  {s_{15}}  {s_{23}}  {s_{45}}^4+4  {s_{12}}^4  {s_{15}}  {s_{23}}  {s_{45}}^4+500  { \epsilon}  {s_{12}}  {s_{15}}^4  {s_{34}}  {s_{45}}^4+125  {s_{12}}  {s_{15}}^4  {s_{34}}  {s_{45}}^4-36  { \epsilon}  {s_{12}}  {s_{23}}^4  {s_{34}}  {s_{45}}^4-9  {s_{12}}  {s_{23}}^4  {s_{34}}  {s_{45}}^4-48  { \epsilon}  {s_{15}}  {s_{23}}^4  {s_{34}}  {s_{45}}^4-12  {s_{15}}  {s_{23}}^4  {s_{34}}  {s_{45}}^4+332  { \epsilon}  {s_{12}}^2  {s_{15}}^3  {s_{34}}  {s_{45}}^4+83  {s_{12}}^2  {s_{15}}^3  {s_{34}}  {s_{45}}^4-324  { \epsilon}  {s_{12}}^2  {s_{23}}^3  {s_{34}}  {s_{45}}^4-81  {s_{12}}^2  {s_{23}}^3  {s_{34}}  {s_{45}}^4-176  { \epsilon}  {s_{15}}^2  {s_{23}}^3  {s_{34}}  {s_{45}}^4-32  {s_{15}}^2  {s_{23}}^3  {s_{34}}  {s_{45}}^4-400  { \epsilon}  {s_{12}}  {s_{15}}  {s_{23}}^3  {s_{34}}  {s_{45}}^4-88  {s_{12}}  {s_{15}}  {s_{23}}^3  {s_{34}}  {s_{45}}^4-668  { \epsilon}  {s_{12}}^3  {s_{15}}^2  {s_{34}}  {s_{45}}^4-167  {s_{12}}^3  {s_{15}}^2  {s_{34}}  {s_{45}}^4-228  { \epsilon}  {s_{12}}^3  {s_{23}}^2  {s_{34}}  {s_{45}}^4-57  {s_{12}}^3  {s_{23}}^2  {s_{34}}  {s_{45}}^4+64  { \epsilon}  {s_{15}}^3  {s_{23}}^2  {s_{34}}  {s_{45}}^4+4  {s_{15}}^3  {s_{23}}^2  {s_{34}}  {s_{45}}^4-1068  { \epsilon}  {s_{12}}  {s_{15}}^2  {s_{23}}^2  {s_{34}}  {s_{45}}^4-219  {s_{12}}  {s_{15}}^2  {s_{23}}^2  {s_{34}}  {s_{45}}^4-664  { \epsilon}  {s_{12}}^2  {s_{15}}  {s_{23}}^2  {s_{34}}  {s_{45}}^4-106  {s_{12}}^2  {s_{15}}  {s_{23}}^2  {s_{34}}  {s_{45}}^4-4  { \epsilon}  {s_{12}}^4  {s_{15}}  {s_{34}}  {s_{45}}^4- {s_{12}}^4  {s_{15}}  {s_{34}}  {s_{45}}^4+4  { \epsilon}  {s_{12}}^4  {s_{23}}  {s_{34}}  {s_{45}}^4+ {s_{12}}^4  {s_{23}}  {s_{34}}  {s_{45}}^4+160  { \epsilon}  {s_{15}}^4  {s_{23}}  {s_{34}}  {s_{45}}^4+40  {s_{15}}^4  {s_{23}}  {s_{34}}  {s_{45}}^4+316  { \epsilon}  {s_{12}}  {s_{15}}^3  {s_{23}}  {s_{34}}  {s_{45}}^4+43  {s_{12}}  {s_{15}}^3  {s_{23}}  {s_{34}}  {s_{45}}^4-1228  { \epsilon}  {s_{12}}^2  {s_{15}}^2  {s_{23}}  {s_{34}}  {s_{45}}^4-307  {s_{12}}^2  {s_{15}}^2  {s_{23}}  {s_{34}}  {s_{45}}^4-300  { \epsilon}  {s_{12}}^3  {s_{15}}  {s_{23}}  {s_{34}}  {s_{45}}^4-39  {s_{12}}^3  {s_{15}}  {s_{23}}  {s_{34}}  {s_{45}}^4+120  { \epsilon}  {s_{12}}^2  {s_{34}}^5  {s_{45}}^3+30  {s_{12}}^2  {s_{34}}^5  {s_{45}}^3+384  { \epsilon}  {s_{23}}^2  {s_{34}}^5  {s_{45}}^3+72  {s_{23}}^2  {s_{34}}^5  {s_{45}}^3-184  { \epsilon}  {s_{12}}  {s_{15}}  {s_{34}}^5  {s_{45}}^3-46  {s_{12}}  {s_{15}}  {s_{34}}^5  {s_{45}}^3+632  { \epsilon}  {s_{12}}  {s_{23}}  {s_{34}}^5  {s_{45}}^3+134  {s_{12}}  {s_{23}}  {s_{34}}^5  {s_{45}}^3-192  { \epsilon}  {s_{15}}  {s_{23}}  {s_{34}}^5  {s_{45}}^3-36  {s_{15}}  {s_{23}}  {s_{34}}^5  {s_{45}}^3+244  { \epsilon}  {s_{12}}^3  {s_{15}}^4  {s_{45}}^3+61  {s_{12}}^3  {s_{15}}^4  {s_{45}}^3+4  { \epsilon}  {s_{12}}^3  {s_{23}}^4  {s_{45}}^3+ {s_{12}}^3  {s_{23}}^4  {s_{45}}^3+124  { \epsilon}  {s_{12}}^2  {s_{15}}  {s_{23}}^4  {s_{45}}^3+31  {s_{12}}^2  {s_{15}}  {s_{23}}^4  {s_{45}}^3-120  { \epsilon}  {s_{12}}^3  {s_{34}}^4  {s_{45}}^3-30  {s_{12}}^3  {s_{34}}^4  {s_{45}}^3-656  { \epsilon}  {s_{23}}^3  {s_{34}}^4  {s_{45}}^3-140  {s_{23}}^3  {s_{34}}^4  {s_{45}}^3+616  { \epsilon}  {s_{12}}  {s_{15}}^2  {s_{34}}^4  {s_{45}}^3+154  {s_{12}}  {s_{15}}^2  {s_{34}}^4  {s_{45}}^3-2256  { \epsilon}  {s_{12}}  {s_{23}}^2  {s_{34}}^4  {s_{45}}^3-444  {s_{12}}  {s_{23}}^2  {s_{34}}^4  {s_{45}}^3-288  { \epsilon}  {s_{15}}  {s_{23}}^2  {s_{34}}^4  {s_{45}}^3-48  {s_{15}}  {s_{23}}^2  {s_{34}}^4  {s_{45}}^3-176  { \epsilon}  {s_{12}}^2  {s_{15}}  {s_{34}}^4  {s_{45}}^3-44  {s_{12}}^2  {s_{15}}  {s_{34}}^4  {s_{45}}^3-1512  { \epsilon}  {s_{12}}^2  {s_{23}}  {s_{34}}^4  {s_{45}}^3-306  {s_{12}}^2  {s_{23}}  {s_{34}}^4  {s_{45}}^3+464  { \epsilon}  {s_{15}}^2  {s_{23}}  {s_{34}}^4  {s_{45}}^3+92  {s_{15}}^2  {s_{23}}  {s_{34}}^4  {s_{45}}^3-504  { \epsilon}  {s_{12}}  {s_{15}}  {s_{23}}  {s_{34}}^4  {s_{45}}^3-150  {s_{12}}  {s_{15}}  {s_{23}}  {s_{34}}^4  {s_{45}}^3-56  { \epsilon}  {s_{12}}^4  {s_{15}}^3  {s_{45}}^3-14  {s_{12}}^4  {s_{15}}^3  {s_{45}}^3+52  { \epsilon}  {s_{12}}^4  {s_{23}}^3  {s_{45}}^3+13  {s_{12}}^4  {s_{23}}^3  {s_{45}}^3-112  { \epsilon}  {s_{12}}^2  {s_{15}}^2  {s_{23}}^3  {s_{45}}^3-28  {s_{12}}^2  {s_{15}}^2  {s_{23}}^3  {s_{45}}^3+360  { \epsilon}  {s_{12}}^3  {s_{15}}  {s_{23}}^3  {s_{45}}^3+90  {s_{12}}^3  {s_{15}}  {s_{23}}^3  {s_{45}}^3+160  { \epsilon}  {s_{23}}^4  {s_{34}}^3  {s_{45}}^3+40  {s_{23}}^4  {s_{34}}^3  {s_{45}}^3-680  { \epsilon}  {s_{12}}  {s_{15}}^3  {s_{34}}^3  {s_{45}}^3-170  {s_{12}}  {s_{15}}^3  {s_{34}}^3  {s_{45}}^3+1940  { \epsilon}  {s_{12}}  {s_{23}}^3  {s_{34}}^3  {s_{45}}^3+437  {s_{12}}  {s_{23}}^3  {s_{34}}^3  {s_{45}}^3+800  { \epsilon}  {s_{15}}  {s_{23}}^3  {s_{34}}^3  {s_{45}}^3+152  {s_{15}}  {s_{23}}^3  {s_{34}}^3  {s_{45}}^3-688  { \epsilon}  {s_{12}}^2  {s_{15}}^2  {s_{34}}^3  {s_{45}}^3-172  {s_{12}}^2  {s_{15}}^2  {s_{34}}^3  {s_{45}}^3+3184  { \epsilon}  {s_{12}}^2  {s_{23}}^2  {s_{34}}^3  {s_{45}}^3+652  {s_{12}}^2  {s_{23}}^2  {s_{34}}^3  {s_{45}}^3-320  { \epsilon}  {s_{15}}^2  {s_{23}}^2  {s_{34}}^3  {s_{45}}^3-56  {s_{15}}^2  {s_{23}}^2  {s_{34}}^3  {s_{45}}^3+2480  { \epsilon}  {s_{12}}  {s_{15}}  {s_{23}}^2  {s_{34}}^3  {s_{45}}^3+452  {s_{12}}  {s_{15}}  {s_{23}}^2  {s_{34}}^3  {s_{45}}^3+544  { \epsilon}  {s_{12}}^3  {s_{15}}  {s_{34}}^3  {s_{45}}^3+136  {s_{12}}^3  {s_{15}}  {s_{34}}^3  {s_{45}}^3+1032  { \epsilon}  {s_{12}}^3  {s_{23}}  {s_{34}}^3  {s_{45}}^3+198  {s_{12}}^3  {s_{23}}  {s_{34}}^3  {s_{45}}^3-352  { \epsilon}  {s_{15}}^3  {s_{23}}  {s_{34}}^3  {s_{45}}^3-76  {s_{15}}^3  {s_{23}}  {s_{34}}^3  {s_{45}}^3-1072  { \epsilon}  {s_{12}}  {s_{15}}^2  {s_{23}}  {s_{34}}^3  {s_{45}}^3-160  {s_{12}}  {s_{15}}^2  {s_{23}}  {s_{34}}^3  {s_{45}}^3+1808  { \epsilon}  {s_{12}}^2  {s_{15}}  {s_{23}}  {s_{34}}^3  {s_{45}}^3+428  {s_{12}}^2  {s_{15}}  {s_{23}}  {s_{34}}^3  {s_{45}}^3+8  { \epsilon}  {s_{12}}^5  {s_{23}}^2  {s_{45}}^3+2  {s_{12}}^5  {s_{23}}^2  {s_{45}}^3-276  { \epsilon}  {s_{12}}^2  {s_{15}}^3  {s_{23}}^2  {s_{45}}^3-69  {s_{12}}^2  {s_{15}}^3  {s_{23}}^2  {s_{45}}^3-496  { \epsilon}  {s_{12}}^3  {s_{15}}^2  {s_{23}}^2  {s_{45}}^3-124  {s_{12}}^3  {s_{15}}^2  {s_{23}}^2  {s_{45}}^3-32  { \epsilon}  {s_{12}}^4  {s_{15}}  {s_{23}}^2  {s_{45}}^3-8  {s_{12}}^4  {s_{15}}  {s_{23}}^2  {s_{45}}^3+248  { \epsilon}  {s_{12}}  {s_{15}}^4  {s_{34}}^2  {s_{45}}^3+62  {s_{12}}  {s_{15}}^4  {s_{34}}^2  {s_{45}}^3-348  { \epsilon}  {s_{12}}  {s_{23}}^4  {s_{34}}^2  {s_{45}}^3-87  {s_{12}}  {s_{23}}^4  {s_{34}}^2  {s_{45}}^3-160  { \epsilon}  {s_{15}}  {s_{23}}^4  {s_{34}}^2  {s_{45}}^3-40  {s_{15}}  {s_{23}}^4  {s_{34}}^2  {s_{45}}^3+1420  { \epsilon}  {s_{12}}^2  {s_{15}}^3  {s_{34}}^2  {s_{45}}^3+355  {s_{12}}^2  {s_{15}}^3  {s_{34}}^2  {s_{45}}^3-1688  { \epsilon}  {s_{12}}^2  {s_{23}}^3  {s_{34}}^2  {s_{45}}^3-398  {s_{12}}^2  {s_{23}}^3  {s_{34}}^2  {s_{45}}^3-144  { \epsilon}  {s_{15}}^2  {s_{23}}^3  {s_{34}}^2  {s_{45}}^3-12  {s_{15}}^2  {s_{23}}^3  {s_{34}}^2  {s_{45}}^3-1336  { \epsilon}  {s_{12}}  {s_{15}}  {s_{23}}^3  {s_{34}}^2  {s_{45}}^3-262  {s_{12}}  {s_{15}}  {s_{23}}^3  {s_{34}}^2  {s_{45}}^3-348  { \epsilon}  {s_{12}}^3  {s_{15}}^2  {s_{34}}^2  {s_{45}}^3-87  {s_{12}}^3  {s_{15}}^2  {s_{34}}^2  {s_{45}}^3-1536  { \epsilon}  {s_{12}}^3  {s_{23}}^2  {s_{34}}^2  {s_{45}}^3-336  {s_{12}}^3  {s_{23}}^2  {s_{34}}^2  {s_{45}}^3+224  { \epsilon}  {s_{15}}^3  {s_{23}}^2  {s_{34}}^2  {s_{45}}^3+32  {s_{15}}^3  {s_{23}}^2  {s_{34}}^2  {s_{45}}^3-808  { \epsilon}  {s_{12}}  {s_{15}}^2  {s_{23}}^2  {s_{34}}^2  {s_{45}}^3-178  {s_{12}}  {s_{15}}^2  {s_{23}}^2  {s_{34}}^2  {s_{45}}^3-2064  { \epsilon}  {s_{12}}^2  {s_{15}}  {s_{23}}^2  {s_{34}}^2  {s_{45}}^3-348  {s_{12}}^2  {s_{15}}  {s_{23}}^2  {s_{34}}^2  {s_{45}}^3-124  { \epsilon}  {s_{12}}^4  {s_{15}}  {s_{34}}^2  {s_{45}}^3-31  {s_{12}}^4  {s_{15}}  {s_{34}}^2  {s_{45}}^3-212  { \epsilon}  {s_{12}}^4  {s_{23}}  {s_{34}}^2  {s_{45}}^3-41  {s_{12}}^4  {s_{23}}  {s_{34}}^2  {s_{45}}^3+80  { \epsilon}  {s_{15}}^4  {s_{23}}  {s_{34}}^2  {s_{45}}^3+20  {s_{15}}^4  {s_{23}}  {s_{34}}^2  {s_{45}}^3+1428  { \epsilon}  {s_{12}}  {s_{15}}^3  {s_{23}}  {s_{34}}^2  {s_{45}}^3+297  {s_{12}}  {s_{15}}^3  {s_{23}}  {s_{34}}^2  {s_{45}}^3-588  { \epsilon}  {s_{12}}^2  {s_{15}}^2  {s_{23}}  {s_{34}}^2  {s_{45}}^3-231  {s_{12}}^2  {s_{15}}^2  {s_{23}}  {s_{34}}^2  {s_{45}}^3-676  { \epsilon}  {s_{12}}^3  {s_{15}}  {s_{23}}  {s_{34}}^2  {s_{45}}^3-133  {s_{12}}^3  {s_{15}}  {s_{23}}  {s_{34}}^2  {s_{45}}^3+264  { \epsilon}  {s_{12}}^2  {s_{15}}^4  {s_{23}}  {s_{45}}^3+66  {s_{12}}^2  {s_{15}}^4  {s_{23}}  {s_{45}}^3-112  { \epsilon}  {s_{12}}^3  {s_{15}}^3  {s_{23}}  {s_{45}}^3-28  {s_{12}}^3  {s_{15}}^3  {s_{23}}  {s_{45}}^3+36  { \epsilon}  {s_{12}}^4  {s_{15}}^2  {s_{23}}  {s_{45}}^3+9  {s_{12}}^4  {s_{15}}^2  {s_{23}}  {s_{45}}^3-8  { \epsilon}  {s_{12}}^5  {s_{15}}  {s_{23}}  {s_{45}}^3-2  {s_{12}}^5  {s_{15}}  {s_{23}}  {s_{45}}^3-676  { \epsilon}  {s_{12}}^2  {s_{15}}^4  {s_{34}}  {s_{45}}^3-169  {s_{12}}^2  {s_{15}}^4  {s_{34}}  {s_{45}}^3+184  { \epsilon}  {s_{12}}^2  {s_{23}}^4  {s_{34}}  {s_{45}}^3+46  {s_{12}}^2  {s_{23}}^4  {s_{34}}  {s_{45}}^3+68  { \epsilon}  {s_{12}}  {s_{15}}  {s_{23}}^4  {s_{34}}  {s_{45}}^3+17  {s_{12}}  {s_{15}}  {s_{23}}^4  {s_{34}}  {s_{45}}^3-320  { \epsilon}  {s_{12}}^3  {s_{15}}^3  {s_{34}}  {s_{45}}^3-80  {s_{12}}^3  {s_{15}}^3  {s_{34}}  {s_{45}}^3+412  { \epsilon}  {s_{12}}^3  {s_{23}}^3  {s_{34}}  {s_{45}}^3+103  {s_{12}}^3  {s_{23}}^3  {s_{34}}  {s_{45}}^3+148  { \epsilon}  {s_{12}}  {s_{15}}^2  {s_{23}}^3  {s_{34}}  {s_{45}}^3+13  {s_{12}}  {s_{15}}^2  {s_{23}}^3  {s_{34}}  {s_{45}}^3-116  { \epsilon}  {s_{12}}^2  {s_{15}}  {s_{23}}^3  {s_{34}}  {s_{45}}^3-53  {s_{12}}^2  {s_{15}}  {s_{23}}^3  {s_{34}}  {s_{45}}^3+180  { \epsilon}  {s_{12}}^4  {s_{15}}^2  {s_{34}}  {s_{45}}^3+45  {s_{12}}^4  {s_{15}}^2  {s_{34}}  {s_{45}}^3+216  { \epsilon}  {s_{12}}^4  {s_{23}}^2  {s_{34}}  {s_{45}}^3+54  {s_{12}}^4  {s_{23}}^2  {s_{34}}  {s_{45}}^3+268  { \epsilon}  {s_{12}}  {s_{15}}^3  {s_{23}}^2  {s_{34}}  {s_{45}}^3+91  {s_{12}}  {s_{15}}^3  {s_{23}}^2  {s_{34}}  {s_{45}}^3+1072  { \epsilon}  {s_{12}}^2  {s_{15}}^2  {s_{23}}^2  {s_{34}}  {s_{45}}^3+244  {s_{12}}^2  {s_{15}}^2  {s_{23}}^2  {s_{34}}  {s_{45}}^3+72  { \epsilon}  {s_{12}}^3  {s_{15}}  {s_{23}}^2  {s_{34}}  {s_{45}}^3-30  {s_{12}}^3  {s_{15}}  {s_{23}}^2  {s_{34}}  {s_{45}}^3-484  { \epsilon}  {s_{12}}  {s_{15}}^4  {s_{23}}  {s_{34}}  {s_{45}}^3-121  {s_{12}}  {s_{15}}^4  {s_{23}}  {s_{34}}  {s_{45}}^3-160  { \epsilon}  {s_{12}}^2  {s_{15}}^3  {s_{23}}  {s_{34}}  {s_{45}}^3-4  {s_{12}}^2  {s_{15}}^3  {s_{23}}  {s_{34}}  {s_{45}}^3+444  { \epsilon}  {s_{12}}^3  {s_{15}}^2  {s_{23}}  {s_{34}}  {s_{45}}^3+135  {s_{12}}^3  {s_{15}}^2  {s_{23}}  {s_{34}}  {s_{45}}^3-92  { \epsilon}  {s_{12}}^4  {s_{15}}  {s_{23}}  {s_{34}}  {s_{45}}^3-35  {s_{12}}^4  {s_{15}}  {s_{23}}  {s_{34}}  {s_{45}}^3-256  { \epsilon}  {s_{23}}^3  {s_{34}}^5  {s_{45}}^2-52  {s_{23}}^3  {s_{34}}^5  {s_{45}}^2-792  { \epsilon}  {s_{12}}  {s_{23}}^2  {s_{34}}^5  {s_{45}}^2-162  {s_{12}}  {s_{23}}^2  {s_{34}}^5  {s_{45}}^2+192  { \epsilon}  {s_{15}}  {s_{23}}^2  {s_{34}}^5  {s_{45}}^2+36  {s_{15}}  {s_{23}}^2  {s_{34}}^5  {s_{45}}^2+120  { \epsilon}  {s_{12}}^2  {s_{15}}  {s_{34}}^5  {s_{45}}^2+30  {s_{12}}^2  {s_{15}}  {s_{34}}^5  {s_{45}}^2-336  { \epsilon}  {s_{12}}^2  {s_{23}}  {s_{34}}^5  {s_{45}}^2-72  {s_{12}}^2  {s_{23}}  {s_{34}}^5  {s_{45}}^2+448  { \epsilon}  {s_{12}}  {s_{15}}  {s_{23}}  {s_{34}}^5  {s_{45}}^2+88  {s_{12}}  {s_{15}}  {s_{23}}  {s_{34}}^5  {s_{45}}^2-60  { \epsilon}  {s_{12}}^4  {s_{15}}^4  {s_{45}}^2-15  {s_{12}}^4  {s_{15}}^4  {s_{45}}^2+8  { \epsilon}  {s_{12}}^4  {s_{23}}^4  {s_{45}}^2+2  {s_{12}}^4  {s_{23}}^4  {s_{45}}^2-52  { \epsilon}  {s_{12}}^3  {s_{15}}  {s_{23}}^4  {s_{45}}^2-13  {s_{12}}^3  {s_{15}}  {s_{23}}^4  {s_{45}}^2+176  { \epsilon}  {s_{23}}^4  {s_{34}}^4  {s_{45}}^2+44  {s_{23}}^4  {s_{34}}^4  {s_{45}}^2+1520  { \epsilon}  {s_{12}}  {s_{23}}^3  {s_{34}}^4  {s_{45}}^2+332  {s_{12}}  {s_{23}}^3  {s_{34}}^4  {s_{45}}^2+208  { \epsilon}  {s_{15}}  {s_{23}}^3  {s_{34}}^4  {s_{45}}^2+28  {s_{15}}  {s_{23}}^3  {s_{34}}^4  {s_{45}}^2-544  { \epsilon}  {s_{12}}^2  {s_{15}}^2  {s_{34}}^4  {s_{45}}^2-136  {s_{12}}^2  {s_{15}}^2  {s_{34}}^4  {s_{45}}^2+2144  { \epsilon}  {s_{12}}^2  {s_{23}}^2  {s_{34}}^4  {s_{45}}^2+440  {s_{12}}^2  {s_{23}}^2  {s_{34}}^4  {s_{45}}^2-320  { \epsilon}  {s_{15}}^2  {s_{23}}^2  {s_{34}}^4  {s_{45}}^2-56  {s_{15}}^2  {s_{23}}^2  {s_{34}}^4  {s_{45}}^2+400  { \epsilon}  {s_{12}}  {s_{15}}  {s_{23}}^2  {s_{34}}^4  {s_{45}}^2+76  {s_{12}}  {s_{15}}  {s_{23}}^2  {s_{34}}^4  {s_{45}}^2+552  { \epsilon}  {s_{12}}^3  {s_{23}}  {s_{34}}^4  {s_{45}}^2+114  {s_{12}}^3  {s_{23}}  {s_{34}}^4  {s_{45}}^2-1064  { \epsilon}  {s_{12}}  {s_{15}}^2  {s_{23}}  {s_{34}}^4  {s_{45}}^2-218  {s_{12}}  {s_{15}}^2  {s_{23}}  {s_{34}}^4  {s_{45}}^2-192  { \epsilon}  {s_{12}}^2  {s_{15}}  {s_{23}}  {s_{34}}^4  {s_{45}}^2-24  { \epsilon}  {s_{12}}^5  {s_{23}}^3  {s_{45}}^2-6  {s_{12}}^5  {s_{23}}^3  {s_{45}}^2+12  { \epsilon}  {s_{12}}^3  {s_{15}}^2  {s_{23}}^3  {s_{45}}^2+3  {s_{12}}^3  {s_{15}}^2  {s_{23}}^3  {s_{45}}^2-132  { \epsilon}  {s_{12}}^4  {s_{15}}  {s_{23}}^3  {s_{45}}^2-33  {s_{12}}^4  {s_{15}}  {s_{23}}^3  {s_{45}}^2-624  { \epsilon}  {s_{12}}  {s_{23}}^4  {s_{34}}^3  {s_{45}}^2-156  {s_{12}}  {s_{23}}^4  {s_{34}}^3  {s_{45}}^2-176  { \epsilon}  {s_{15}}  {s_{23}}^4  {s_{34}}^3  {s_{45}}^2-44  {s_{15}}  {s_{23}}^4  {s_{34}}^3  {s_{45}}^2+728  { \epsilon}  {s_{12}}^2  {s_{15}}^3  {s_{34}}^3  {s_{45}}^2+182  {s_{12}}^2  {s_{15}}^3  {s_{34}}^3  {s_{45}}^2-2500  { \epsilon}  {s_{12}}^2  {s_{23}}^3  {s_{34}}^3  {s_{45}}^2-577  {s_{12}}^2  {s_{23}}^3  {s_{34}}^3  {s_{45}}^2+48  { \epsilon}  {s_{15}}^2  {s_{23}}^3  {s_{34}}^3  {s_{45}}^2+24  {s_{15}}^2  {s_{23}}^3  {s_{34}}^3  {s_{45}}^2-1028  { \epsilon}  {s_{12}}  {s_{15}}  {s_{23}}^3  {s_{34}}^3  {s_{45}}^2-173  {s_{12}}  {s_{15}}  {s_{23}}^3  {s_{34}}^3  {s_{45}}^2+428  { \epsilon}  {s_{12}}^3  {s_{15}}^2  {s_{34}}^3  {s_{45}}^2+107  {s_{12}}^3  {s_{15}}^2  {s_{34}}^3  {s_{45}}^2-1980  { \epsilon}  {s_{12}}^3  {s_{23}}^2  {s_{34}}^3  {s_{45}}^2-423  {s_{12}}^3  {s_{23}}^2  {s_{34}}^3  {s_{45}}^2+128  { \epsilon}  {s_{15}}^3  {s_{23}}^2  {s_{34}}^3  {s_{45}}^2+20  {s_{15}}^3  {s_{23}}^2  {s_{34}}^3  {s_{45}}^2+356  { \epsilon}  {s_{12}}  {s_{15}}^2  {s_{23}}^2  {s_{34}}^3  {s_{45}}^2+41  {s_{12}}  {s_{15}}^2  {s_{23}}^2  {s_{34}}^3  {s_{45}}^2-748  { \epsilon}  {s_{12}}^2  {s_{15}}  {s_{23}}^2  {s_{34}}^3  {s_{45}}^2-79  {s_{12}}^2  {s_{15}}  {s_{23}}^2  {s_{34}}^3  {s_{45}}^2-120  { \epsilon}  {s_{12}}^4  {s_{15}}  {s_{34}}^3  {s_{45}}^2-30  {s_{12}}^4  {s_{15}}  {s_{34}}^3  {s_{45}}^2-216  { \epsilon}  {s_{12}}^4  {s_{23}}  {s_{34}}^3  {s_{45}}^2-42  {s_{12}}^4  {s_{23}}  {s_{34}}^3  {s_{45}}^2+848  { \epsilon}  {s_{12}}  {s_{15}}^3  {s_{23}}  {s_{34}}^3  {s_{45}}^2+188  {s_{12}}  {s_{15}}^3  {s_{23}}  {s_{34}}^3  {s_{45}}^2+876  { \epsilon}  {s_{12}}^2  {s_{15}}^2  {s_{23}}  {s_{34}}^3  {s_{45}}^2+111  {s_{12}}^2  {s_{15}}^2  {s_{23}}  {s_{34}}^3  {s_{45}}^2+20  { \epsilon}  {s_{12}}^3  {s_{15}}  {s_{23}}  {s_{34}}^3  {s_{45}}^2-19  {s_{12}}^3  {s_{15}}  {s_{23}}  {s_{34}}^3  {s_{45}}^2+132  { \epsilon}  {s_{12}}^3  {s_{15}}^3  {s_{23}}^2  {s_{45}}^2+33  {s_{12}}^3  {s_{15}}^3  {s_{23}}^2  {s_{45}}^2+180  { \epsilon}  {s_{12}}^4  {s_{15}}^2  {s_{23}}^2  {s_{45}}^2+45  {s_{12}}^4  {s_{15}}^2  {s_{23}}^2  {s_{45}}^2+48  { \epsilon}  {s_{12}}^5  {s_{15}}  {s_{23}}^2  {s_{45}}^2+12  {s_{12}}^5  {s_{15}}  {s_{23}}^2  {s_{45}}^2-304  { \epsilon}  {s_{12}}^2  {s_{15}}^4  {s_{34}}^2  {s_{45}}^2-76  {s_{12}}^2  {s_{15}}^4  {s_{34}}^2  {s_{45}}^2+668  { \epsilon}  {s_{12}}^2  {s_{23}}^4  {s_{34}}^2  {s_{45}}^2+167  {s_{12}}^2  {s_{23}}^4  {s_{34}}^2  {s_{45}}^2+388  { \epsilon}  {s_{12}}  {s_{15}}  {s_{23}}^4  {s_{34}}^2  {s_{45}}^2+97  {s_{12}}  {s_{15}}  {s_{23}}^4  {s_{34}}^2  {s_{45}}^2-792  { \epsilon}  {s_{12}}^3  {s_{15}}^3  {s_{34}}^2  {s_{45}}^2-198  {s_{12}}^3  {s_{15}}^3  {s_{34}}^2  {s_{45}}^2+1536  { \epsilon}  {s_{12}}^3  {s_{23}}^3  {s_{34}}^2  {s_{45}}^2+372  {s_{12}}^3  {s_{23}}^3  {s_{34}}^2  {s_{45}}^2-128  { \epsilon}  {s_{12}}  {s_{15}}^2  {s_{23}}^3  {s_{34}}^2  {s_{45}}^2-68  {s_{12}}  {s_{15}}^2  {s_{23}}^3  {s_{34}}^2  {s_{45}}^2+528  { \epsilon}  {s_{12}}^2  {s_{15}}  {s_{23}}^3  {s_{34}}^2  {s_{45}}^2+72  {s_{12}}^2  {s_{15}}  {s_{23}}^3  {s_{34}}^2  {s_{45}}^2+56  { \epsilon}  {s_{12}}^4  {s_{15}}^2  {s_{34}}^2  {s_{45}}^2+14  {s_{12}}^4  {s_{15}}^2  {s_{34}}^2  {s_{45}}^2+732  { \epsilon}  {s_{12}}^4  {s_{23}}^2  {s_{34}}^2  {s_{45}}^2+171  {s_{12}}^4  {s_{23}}^2  {s_{34}}^2  {s_{45}}^2-28  { \epsilon}  {s_{12}}  {s_{15}}^3  {s_{23}}^2  {s_{34}}^2  {s_{45}}^2+29  {s_{12}}  {s_{15}}^3  {s_{23}}^2  {s_{34}}^2  {s_{45}}^2+224  { \epsilon}  {s_{12}}^2  {s_{15}}^2  {s_{23}}^2  {s_{34}}^2  {s_{45}}^2+56  {s_{12}}^2  {s_{15}}^2  {s_{23}}^2  {s_{34}}^2  {s_{45}}^2-240  { \epsilon}  {s_{12}}^3  {s_{15}}  {s_{23}}^2  {s_{34}}^2  {s_{45}}^2-132  {s_{12}}^3  {s_{15}}  {s_{23}}^2  {s_{34}}^2  {s_{45}}^2-232  { \epsilon}  {s_{12}}  {s_{15}}^4  {s_{23}}  {s_{34}}^2  {s_{45}}^2-58  {s_{12}}  {s_{15}}^4  {s_{23}}  {s_{34}}^2  {s_{45}}^2-828  { \epsilon}  {s_{12}}^2  {s_{15}}^3  {s_{23}}  {s_{34}}^2  {s_{45}}^2-159  {s_{12}}^2  {s_{15}}^3  {s_{23}}  {s_{34}}^2  {s_{45}}^2+192  { \epsilon}  {s_{12}}^3  {s_{15}}^2  {s_{23}}  {s_{34}}^2  {s_{45}}^2+108  {s_{12}}^3  {s_{15}}^2  {s_{23}}  {s_{34}}^2  {s_{45}}^2-380  { \epsilon}  {s_{12}}^4  {s_{15}}  {s_{23}}  {s_{34}}^2  {s_{45}}^2-95  {s_{12}}^4  {s_{15}}  {s_{23}}  {s_{34}}^2  {s_{45}}^2-92  { \epsilon}  {s_{12}}^3  {s_{15}}^4  {s_{23}}  {s_{45}}^2-23  {s_{12}}^3  {s_{15}}^4  {s_{23}}  {s_{45}}^2+4  { \epsilon}  {s_{12}}^4  {s_{15}}^3  {s_{23}}  {s_{45}}^2+ {s_{12}}^4  {s_{15}}^3  {s_{23}}  {s_{45}}^2-24  { \epsilon}  {s_{12}}^5  {s_{15}}^2  {s_{23}}  {s_{45}}^2-6  {s_{12}}^5  {s_{15}}^2  {s_{23}}  {s_{45}}^2+364  { \epsilon}  {s_{12}}^3  {s_{15}}^4  {s_{34}}  {s_{45}}^2+91  {s_{12}}^3  {s_{15}}^4  {s_{34}}  {s_{45}}^2-228  { \epsilon}  {s_{12}}^3  {s_{23}}^4  {s_{34}}  {s_{45}}^2-57  {s_{12}}^3  {s_{23}}^4  {s_{34}}  {s_{45}}^2-40  { \epsilon}  {s_{12}}^2  {s_{15}}  {s_{23}}^4  {s_{34}}  {s_{45}}^2-10  {s_{12}}^2  {s_{15}}  {s_{23}}^4  {s_{34}}  {s_{45}}^2+124  { \epsilon}  {s_{12}}^4  {s_{15}}^3  {s_{34}}  {s_{45}}^2+31  {s_{12}}^4  {s_{15}}^3  {s_{34}}  {s_{45}}^2-276  { \epsilon}  {s_{12}}^4  {s_{23}}^3  {s_{34}}  {s_{45}}^2-69  {s_{12}}^4  {s_{23}}^3  {s_{34}}  {s_{45}}^2+160  { \epsilon}  {s_{12}}^2  {s_{15}}^2  {s_{23}}^3  {s_{34}}  {s_{45}}^2+52  {s_{12}}^2  {s_{15}}^2  {s_{23}}^3  {s_{34}}  {s_{45}}^2+448  { \epsilon}  {s_{12}}^3  {s_{15}}  {s_{23}}^3  {s_{34}}  {s_{45}}^2+124  {s_{12}}^3  {s_{15}}  {s_{23}}^3  {s_{34}}  {s_{45}}^2-104  { \epsilon}  {s_{12}}^5  {s_{23}}^2  {s_{34}}  {s_{45}}^2-26  {s_{12}}^5  {s_{23}}^2  {s_{34}}  {s_{45}}^2-536  { \epsilon}  {s_{12}}^2  {s_{15}}^3  {s_{23}}^2  {s_{34}}  {s_{45}}^2-146  {s_{12}}^2  {s_{15}}^3  {s_{23}}^2  {s_{34}}  {s_{45}}^2-504  { \epsilon}  {s_{12}}^3  {s_{15}}^2  {s_{23}}^2  {s_{34}}  {s_{45}}^2-126  {s_{12}}^3  {s_{15}}^2  {s_{23}}^2  {s_{34}}  {s_{45}}^2+252  { \epsilon}  {s_{12}}^4  {s_{15}}  {s_{23}}^2  {s_{34}}  {s_{45}}^2+75  {s_{12}}^4  {s_{15}}  {s_{23}}^2  {s_{34}}  {s_{45}}^2+416  { \epsilon}  {s_{12}}^2  {s_{15}}^4  {s_{23}}  {s_{34}}  {s_{45}}^2+104  {s_{12}}^2  {s_{15}}^4  {s_{23}}  {s_{34}}  {s_{45}}^2-80  { \epsilon}  {s_{12}}^3  {s_{15}}^3  {s_{23}}  {s_{34}}  {s_{45}}^2-32  {s_{12}}^3  {s_{15}}^3  {s_{23}}  {s_{34}}  {s_{45}}^2-100  { \epsilon}  {s_{12}}^4  {s_{15}}^2  {s_{23}}  {s_{34}}  {s_{45}}^2-37  {s_{12}}^4  {s_{15}}^2  {s_{23}}  {s_{34}}  {s_{45}}^2+104  { \epsilon}  {s_{12}}^5  {s_{15}}  {s_{23}}  {s_{34}}  {s_{45}}^2+26  {s_{12}}^5  {s_{15}}  {s_{23}}  {s_{34}}  {s_{45}}^2+64  { \epsilon}  {s_{23}}^4  {s_{34}}^5  {s_{45}}+16  {s_{23}}^4  {s_{34}}^5  {s_{45}}+440  { \epsilon}  {s_{12}}  {s_{23}}^3  {s_{34}}^5  {s_{45}}+98  {s_{12}}  {s_{23}}^3  {s_{34}}^5  {s_{45}}-64  { \epsilon}  {s_{15}}  {s_{23}}^3  {s_{34}}^5  {s_{45}}-16  {s_{15}}  {s_{23}}^3  {s_{34}}^5  {s_{45}}+312  { \epsilon}  {s_{12}}^2  {s_{23}}^2  {s_{34}}^5  {s_{45}}+66  {s_{12}}^2  {s_{23}}^2  {s_{34}}^5  {s_{45}}-344  { \epsilon}  {s_{12}}  {s_{15}}  {s_{23}}^2  {s_{34}}^5  {s_{45}}-74  {s_{12}}  {s_{15}}  {s_{23}}^2  {s_{34}}^5  {s_{45}}-216  { \epsilon}  {s_{12}}^2  {s_{15}}  {s_{23}}  {s_{34}}^5  {s_{45}}-42  {s_{12}}^2  {s_{15}}  {s_{23}}  {s_{34}}^5  {s_{45}}-408  { \epsilon}  {s_{12}}  {s_{23}}^4  {s_{34}}^4  {s_{45}}-102  {s_{12}}  {s_{23}}^4  {s_{34}}^4  {s_{45}}-64  { \epsilon}  {s_{15}}  {s_{23}}^4  {s_{34}}^4  {s_{45}}-16  {s_{15}}  {s_{23}}^4  {s_{34}}^4  {s_{45}}-1288  { \epsilon}  {s_{12}}^2  {s_{23}}^3  {s_{34}}^4  {s_{45}}-298  {s_{12}}^2  {s_{23}}^3  {s_{34}}^4  {s_{45}}+64  { \epsilon}  {s_{15}}^2  {s_{23}}^3  {s_{34}}^4  {s_{45}}+16  {s_{15}}^2  {s_{23}}^3  {s_{34}}^4  {s_{45}}-152  { \epsilon}  {s_{12}}  {s_{15}}  {s_{23}}^3  {s_{34}}^4  {s_{45}}-14  {s_{12}}  {s_{15}}  {s_{23}}^3  {s_{34}}^4  {s_{45}}+120  { \epsilon}  {s_{12}}^3  {s_{15}}^2  {s_{34}}^4  {s_{45}}+30  {s_{12}}^3  {s_{15}}^2  {s_{34}}^4  {s_{45}}-720  { \epsilon}  {s_{12}}^3  {s_{23}}^2  {s_{34}}^4  {s_{45}}-156  {s_{12}}^3  {s_{23}}^2  {s_{34}}^4  {s_{45}}+464  { \epsilon}  {s_{12}}  {s_{15}}^2  {s_{23}}^2  {s_{34}}^4  {s_{45}}+92  {s_{12}}  {s_{15}}^2  {s_{23}}^2  {s_{34}}^4  {s_{45}}+576  { \epsilon}  {s_{12}}^2  {s_{15}}  {s_{23}}^2  {s_{34}}^4  {s_{45}}+144  {s_{12}}^2  {s_{15}}  {s_{23}}^2  {s_{34}}^4  {s_{45}}+424  { \epsilon}  {s_{12}}^2  {s_{15}}^2  {s_{23}}  {s_{34}}^4  {s_{45}}+82  {s_{12}}^2  {s_{15}}^2  {s_{23}}  {s_{34}}^4  {s_{45}}+408  { \epsilon}  {s_{12}}^3  {s_{15}}  {s_{23}}  {s_{34}}^4  {s_{45}}+78  {s_{12}}^3  {s_{15}}  {s_{23}}  {s_{34}}^4  {s_{45}}+744  { \epsilon}  {s_{12}}^2  {s_{23}}^4  {s_{34}}^3  {s_{45}}+186  {s_{12}}^2  {s_{23}}^4  {s_{34}}^3  {s_{45}}+344  { \epsilon}  {s_{12}}  {s_{15}}  {s_{23}}^4  {s_{34}}^3  {s_{45}}+86  {s_{12}}  {s_{15}}  {s_{23}}^4  {s_{34}}^3  {s_{45}}-240  { \epsilon}  {s_{12}}^3  {s_{15}}^3  {s_{34}}^3  {s_{45}}-60  {s_{12}}^3  {s_{15}}^3  {s_{34}}^3  {s_{45}}+1344  { \epsilon}  {s_{12}}^3  {s_{23}}^3  {s_{34}}^3  {s_{45}}+324  {s_{12}}^3  {s_{23}}^3  {s_{34}}^3  {s_{45}}-224  { \epsilon}  {s_{12}}  {s_{15}}^2  {s_{23}}^3  {s_{34}}^3  {s_{45}}-68  {s_{12}}  {s_{15}}^2  {s_{23}}^3  {s_{34}}^3  {s_{45}}+36  { \epsilon}  {s_{12}}^2  {s_{15}}  {s_{23}}^3  {s_{34}}^3  {s_{45}}-27  {s_{12}}^2  {s_{15}}  {s_{23}}^3  {s_{34}}^3  {s_{45}}-120  { \epsilon}  {s_{12}}^4  {s_{15}}^2  {s_{34}}^3  {s_{45}}-30  {s_{12}}^4  {s_{15}}^2  {s_{34}}^3  {s_{45}}+504  { \epsilon}  {s_{12}}^4  {s_{23}}^2  {s_{34}}^3  {s_{45}}+114  {s_{12}}^4  {s_{23}}^2  {s_{34}}^3  {s_{45}}-120  { \epsilon}  {s_{12}}  {s_{15}}^3  {s_{23}}^2  {s_{34}}^3  {s_{45}}-18  {s_{12}}  {s_{15}}^3  {s_{23}}^2  {s_{34}}^3  {s_{45}}-324  { \epsilon}  {s_{12}}^2  {s_{15}}^2  {s_{23}}^2  {s_{34}}^3  {s_{45}}-57  {s_{12}}^2  {s_{15}}^2  {s_{23}}^2  {s_{34}}^3  {s_{45}}-788  { \epsilon}  {s_{12}}^3  {s_{15}}  {s_{23}}^2  {s_{34}}^3  {s_{45}}-221  {s_{12}}^3  {s_{15}}  {s_{23}}^2  {s_{34}}^3  {s_{45}}-360  { \epsilon}  {s_{12}}^2  {s_{15}}^3  {s_{23}}  {s_{34}}^3  {s_{45}}-78  {s_{12}}^2  {s_{15}}^3  {s_{23}}  {s_{34}}^3  {s_{45}}-124  { \epsilon}  {s_{12}}^3  {s_{15}}^2  {s_{23}}  {s_{34}}^3  {s_{45}}+5  {s_{12}}^3  {s_{15}}^2  {s_{23}}  {s_{34}}^3  {s_{45}}-288  { \epsilon}  {s_{12}}^4  {s_{15}}  {s_{23}}  {s_{34}}^3  {s_{45}}-60  {s_{12}}^4  {s_{15}}  {s_{23}}  {s_{34}}^3  {s_{45}}+120  { \epsilon}  {s_{12}}^3  {s_{15}}^4  {s_{34}}^2  {s_{45}}+30  {s_{12}}^3  {s_{15}}^4  {s_{34}}^2  {s_{45}}-520  { \epsilon}  {s_{12}}^3  {s_{23}}^4  {s_{34}}^2  {s_{45}}-130  {s_{12}}^3  {s_{23}}^4  {s_{34}}^2  {s_{45}}-340  { \epsilon}  {s_{12}}^2  {s_{15}}  {s_{23}}^4  {s_{34}}^2  {s_{45}}-85  {s_{12}}^2  {s_{15}}  {s_{23}}^4  {s_{34}}^2  {s_{45}}+180  { \epsilon}  {s_{12}}^4  {s_{15}}^3  {s_{34}}^2  {s_{45}}+45  {s_{12}}^4  {s_{15}}^3  {s_{34}}^2  {s_{45}}-584  { \epsilon}  {s_{12}}^4  {s_{23}}^3  {s_{34}}^2  {s_{45}}-146  {s_{12}}^4  {s_{23}}^3  {s_{34}}^2  {s_{45}}+528  { \epsilon}  {s_{12}}^2  {s_{15}}^2  {s_{23}}^3  {s_{34}}^2  {s_{45}}+144  {s_{12}}^2  {s_{15}}^2  {s_{23}}^3  {s_{34}}^2  {s_{45}}+424  { \epsilon}  {s_{12}}^3  {s_{15}}  {s_{23}}^3  {s_{34}}^2  {s_{45}}+118  {s_{12}}^3  {s_{15}}  {s_{23}}^3  {s_{34}}^2  {s_{45}}-96  { \epsilon}  {s_{12}}^5  {s_{23}}^2  {s_{34}}^2  {s_{45}}-24  {s_{12}}^5  {s_{23}}^2  {s_{34}}^2  {s_{45}}-340  { \epsilon}  {s_{12}}^2  {s_{15}}^3  {s_{23}}^2  {s_{34}}^2  {s_{45}}-97  {s_{12}}^2  {s_{15}}^3  {s_{23}}^2  {s_{34}}^2  {s_{45}}+8  { \epsilon}  {s_{12}}^3  {s_{15}}^2  {s_{23}}^2  {s_{34}}^2  {s_{45}}+2  {s_{12}}^3  {s_{15}}^2  {s_{23}}^2  {s_{34}}^2  {s_{45}}+732  { \epsilon}  {s_{12}}^4  {s_{15}}  {s_{23}}^2  {s_{34}}^2  {s_{45}}+195  {s_{12}}^4  {s_{15}}  {s_{23}}^2  {s_{34}}^2  {s_{45}}+152  { \epsilon}  {s_{12}}^2  {s_{15}}^4  {s_{23}}  {s_{34}}^2  {s_{45}}+38  {s_{12}}^2  {s_{15}}^4  {s_{23}}  {s_{34}}^2  {s_{45}}-32  { \epsilon}  {s_{12}}^3  {s_{15}}^3  {s_{23}}  {s_{34}}^2  {s_{45}}-20  {s_{12}}^3  {s_{15}}^3  {s_{23}}  {s_{34}}^2  {s_{45}}-328  { \epsilon}  {s_{12}}^4  {s_{15}}^2  {s_{23}}  {s_{34}}^2  {s_{45}}-94  {s_{12}}^4  {s_{15}}^2  {s_{23}}  {s_{34}}^2  {s_{45}}+96  { \epsilon}  {s_{12}}^5  {s_{15}}  {s_{23}}  {s_{34}}^2  {s_{45}}+24  {s_{12}}^5  {s_{15}}  {s_{23}}  {s_{34}}^2  {s_{45}}-60  { \epsilon}  {s_{12}}^4  {s_{15}}^4  {s_{34}}  {s_{45}}-15  {s_{12}}^4  {s_{15}}^4  {s_{34}}  {s_{45}}+120  { \epsilon}  {s_{12}}^4  {s_{23}}^4  {s_{34}}  {s_{45}}+30  {s_{12}}^4  {s_{23}}^4  {s_{34}}  {s_{45}}+60  { \epsilon}  {s_{12}}^3  {s_{15}}  {s_{23}}^4  {s_{34}}  {s_{45}}+15  {s_{12}}^3  {s_{15}}  {s_{23}}^4  {s_{34}}  {s_{45}}+88  { \epsilon}  {s_{12}}^5  {s_{23}}^3  {s_{34}}  {s_{45}}+22  {s_{12}}^5  {s_{23}}^3  {s_{34}}  {s_{45}}-212  { \epsilon}  {s_{12}}^3  {s_{15}}^2  {s_{23}}^3  {s_{34}}  {s_{45}}-53  {s_{12}}^3  {s_{15}}^2  {s_{23}}^3  {s_{34}}  {s_{45}}-244  { \epsilon}  {s_{12}}^4  {s_{15}}  {s_{23}}^3  {s_{34}}  {s_{45}}-61  {s_{12}}^4  {s_{15}}  {s_{23}}^3  {s_{34}}  {s_{45}}+244  { \epsilon}  {s_{12}}^3  {s_{15}}^3  {s_{23}}^2  {s_{34}}  {s_{45}}+61  {s_{12}}^3  {s_{15}}^3  {s_{23}}^2  {s_{34}}  {s_{45}}+68  { \epsilon}  {s_{12}}^4  {s_{15}}^2  {s_{23}}^2  {s_{34}}  {s_{45}}+17  {s_{12}}^4  {s_{15}}^2  {s_{23}}^2  {s_{34}}  {s_{45}}-176  { \epsilon}  {s_{12}}^5  {s_{15}}  {s_{23}}^2  {s_{34}}  {s_{45}}-44  {s_{12}}^5  {s_{15}}  {s_{23}}^2  {s_{34}}  {s_{45}}-92  { \epsilon}  {s_{12}}^3  {s_{15}}^4  {s_{23}}  {s_{34}}  {s_{45}}-23  {s_{12}}^3  {s_{15}}^4  {s_{23}}  {s_{34}}  {s_{45}}+116  { \epsilon}  {s_{12}}^4  {s_{15}}^3  {s_{23}}  {s_{34}}  {s_{45}}+29  {s_{12}}^4  {s_{15}}^3  {s_{23}}  {s_{34}}  {s_{45}}+88  { \epsilon}  {s_{12}}^5  {s_{15}}^2  {s_{23}}  {s_{34}}  {s_{45}}+22  {s_{12}}^5  {s_{15}}^2  {s_{23}}  {s_{34}}  {s_{45}}-96  { \epsilon}  {s_{12}}  {s_{23}}^4  {s_{34}}^5-24  {s_{12}}  {s_{23}}^4  {s_{34}}^5-96  { \epsilon}  {s_{12}}^2  {s_{23}}^3  {s_{34}}^5-24  {s_{12}}^2  {s_{23}}^3  {s_{34}}^5+96  { \epsilon}  {s_{12}}  {s_{15}}  {s_{23}}^3  {s_{34}}^5+24  {s_{12}}  {s_{15}}  {s_{23}}^3  {s_{34}}^5+96  { \epsilon}  {s_{12}}^2  {s_{15}}  {s_{23}}^2  {s_{34}}^5+24  {s_{12}}^2  {s_{15}}  {s_{23}}^2  {s_{34}}^5+288  { \epsilon}  {s_{12}}^2  {s_{23}}^4  {s_{34}}^4+72  {s_{12}}^2  {s_{23}}^4  {s_{34}}^4+96  { \epsilon}  {s_{12}}  {s_{15}}  {s_{23}}^4  {s_{34}}^4+24  {s_{12}}  {s_{15}}  {s_{23}}^4  {s_{34}}^4+288  { \epsilon}  {s_{12}}^3  {s_{23}}^3  {s_{34}}^4+72  {s_{12}}^3  {s_{23}}^3  {s_{34}}^4-96  { \epsilon}  {s_{12}}  {s_{15}}^2  {s_{23}}^3  {s_{34}}^4-24  {s_{12}}  {s_{15}}^2  {s_{23}}^3  {s_{34}}^4-288  { \epsilon}  {s_{12}}^2  {s_{15}}  {s_{23}}^3  {s_{34}}^4-72  {s_{12}}^2  {s_{15}}  {s_{23}}^3  {s_{34}}^4-384  { \epsilon}  {s_{12}}^3  {s_{15}}  {s_{23}}^2  {s_{34}}^4-96  {s_{12}}^3  {s_{15}}  {s_{23}}^2  {s_{34}}^4+96  { \epsilon}  {s_{12}}^3  {s_{15}}^2  {s_{23}}  {s_{34}}^4+24  {s_{12}}^3  {s_{15}}^2  {s_{23}}  {s_{34}}^4-288  { \epsilon}  {s_{12}}^3  {s_{23}}^4  {s_{34}}^3-72  {s_{12}}^3  {s_{23}}^4  {s_{34}}^3-192  { \epsilon}  {s_{12}}^2  {s_{15}}  {s_{23}}^4  {s_{34}}^3-48  {s_{12}}^2  {s_{15}}  {s_{23}}^4  {s_{34}}^3-288  { \epsilon}  {s_{12}}^4  {s_{23}}^3  {s_{34}}^3-72  {s_{12}}^4  {s_{23}}^3  {s_{34}}^3+288  { \epsilon}  {s_{12}}^2  {s_{15}}^2  {s_{23}}^3  {s_{34}}^3+72  {s_{12}}^2  {s_{15}}^2  {s_{23}}^3  {s_{34}}^3+288  { \epsilon}  {s_{12}}^3  {s_{15}}  {s_{23}}^3  {s_{34}}^3+72  {s_{12}}^3  {s_{15}}  {s_{23}}^3  {s_{34}}^3-96  { \epsilon}  {s_{12}}^2  {s_{15}}^3  {s_{23}}^2  {s_{34}}^3-24  {s_{12}}^2  {s_{15}}^3  {s_{23}}^2  {s_{34}}^3+96  { \epsilon}  {s_{12}}^3  {s_{15}}^2  {s_{23}}^2  {s_{34}}^3+24  {s_{12}}^3  {s_{15}}^2  {s_{23}}^2  {s_{34}}^3+480  { \epsilon}  {s_{12}}^4  {s_{15}}  {s_{23}}^2  {s_{34}}^3+120  {s_{12}}^4  {s_{15}}  {s_{23}}^2  {s_{34}}^3-96  { \epsilon}  {s_{12}}^3  {s_{15}}^3  {s_{23}}  {s_{34}}^3-24  {s_{12}}^3  {s_{15}}^3  {s_{23}}  {s_{34}}^3-192  { \epsilon}  {s_{12}}^4  {s_{15}}^2  {s_{23}}  {s_{34}}^3-48  {s_{12}}^4  {s_{15}}^2  {s_{23}}  {s_{34}}^3+96  { \epsilon}  {s_{12}}^4  {s_{23}}^4  {s_{34}}^2+24  {s_{12}}^4  {s_{23}}^4  {s_{34}}^2+96  { \epsilon}  {s_{12}}^3  {s_{15}}  {s_{23}}^4  {s_{34}}^2+24  {s_{12}}^3  {s_{15}}  {s_{23}}^4  {s_{34}}^2+96  { \epsilon}  {s_{12}}^5  {s_{23}}^3  {s_{34}}^2+24  {s_{12}}^5  {s_{23}}^3  {s_{34}}^2-192  { \epsilon}  {s_{12}}^3  {s_{15}}^2  {s_{23}}^3  {s_{34}}^2-48  {s_{12}}^3  {s_{15}}^2  {s_{23}}^3  {s_{34}}^2-96  { \epsilon}  {s_{12}}^4  {s_{15}}  {s_{23}}^3  {s_{34}}^2-24  {s_{12}}^4  {s_{15}}  {s_{23}}^3  {s_{34}}^2+96  { \epsilon}  {s_{12}}^3  {s_{15}}^3  {s_{23}}^2  {s_{34}}^2+24  {s_{12}}^3  {s_{15}}^3  {s_{23}}^2  {s_{34}}^2-96  { \epsilon}  {s_{12}}^4  {s_{15}}^2  {s_{23}}^2  {s_{34}}^2-24  {s_{12}}^4  {s_{15}}^2  {s_{23}}^2  {s_{34}}^2-192  { \epsilon}  {s_{12}}^5  {s_{15}}  {s_{23}}^2  {s_{34}}^2-48  {s_{12}}^5  {s_{15}}  {s_{23}}^2  {s_{34}}^2+96  { \epsilon}  {s_{12}}^4  {s_{15}}^3  {s_{23}}  {s_{34}}^2+24  {s_{12}}^4  {s_{15}}^3  {s_{23}}  {s_{34}}^2+96  { \epsilon}  {s_{12}}^5  {s_{15}}^2  {s_{23}}  {s_{34}}^2+24  {s_{12}}^5  {s_{15}}^2  {s_{23}}  {s_{34}}^2\big)/\;\big(8 (4  { \epsilon}+1)  {s_{12}}  {s_{23}}  {s_{45}}^2 ( {s_{12}}- {s_{45}}) ( {s_{34}}+ {s_{45}})( {s_{12}}+ {s_{15}}- {s_{34}})( {s_{12}}+ {s_{23}}- {s_{45}})\\ ( {s_{12}}- {s_{34}}- {s_{45}}) (- {s_{15}}+ {s_{23}}+ {s_{34}})(- {s_{15}}+ {s_{23}}- {s_{45}})\big)
\end{dmath}
After running our algorithm, $c_{107}$ is reduced to about $9$ lines long:
\begin{align}
\begin{aligned}
c_{107}&=\frac{3  {s_{23}}  {s_{34}}}{2 (4  { \epsilon}+1)  {s_{12}}  {s_{45}} (- {s_{15}}+ {s_{23}}+ {s_{34}})}-\frac{3  {s_{34}}}{2 (4  { \epsilon}+1)  {s_{12}}  {s_{45}}}+\frac{15  {s_{15}}^2-15  {s_{15}}  {s_{34}}}{ 8{s_{23}}  {s_{45}} (- {s_{12}}- {s_{15}}+ {s_{34}})}\\
&+\frac{ {s_{23}}  {s_{34}}^2}{ {s_{45}}^2 ( {s_{45}}- {s_{12}}) (- {s_{15}}+ {s_{23}}+ {s_{34}})}-\frac{2  {s_{23}}  {s_{34}}^2}{ {s_{12}}  {s_{45}}^2 (- {s_{15}}+ {s_{23}}+ {s_{34}})}+\frac{ {s_{23}}  {s_{34}}+ {s_{34}}^2}{ {s_{45}} ( {s_{45}}- {s_{12}}) (- {s_{15}}+ {s_{23}}+ {s_{34}})}\\
&-\frac{11  {s_{23}}  {s_{34}}}{2  {s_{12}}  {s_{45}} (- {s_{15}}+ {s_{23}}+ {s_{34}})}-\frac{15  {s_{15}}  {s_{34}}}{8  {s_{23}}  {s_{45}} (- {s_{12}}+ {s_{34}}+ {s_{45}})}+\frac{ {s_{15}}- {s_{23}}- {s_{34}}}{ {s_{45}} (- {s_{12}}- {s_{23}}+ {s_{45}})}+\frac{2  {s_{15}}-2  {s_{34}}}{ {s_{12}}  {s_{23}}}\\
&+\frac{15  {s_{15}}-15  {s_{34}}}{ 8{s_{23}} (- {s_{12}}- {s_{15}}+ {s_{34}})}-\frac{7  {s_{23}}}{2  {s_{12}} (- {s_{15}}+ {s_{23}}+ {s_{34}})}-\frac{15  {s_{23}}}{4 (- {s_{12}}- {s_{15}}+ {s_{34}}) (- {s_{15}}+ {s_{23}}+ {s_{34}})}\\
&-\frac{ {s_{15}}}{2  {s_{12}} (- {s_{12}}- {s_{23}}+ {s_{45}})}+\frac{ {s_{23}}-{s_{45}}}{ 2{s_{12}} ( {s_{15}}- {s_{23}}+ {s_{45}})}+\frac{15  {s_{15}}}{8  {s_{45}} (- {s_{12}}- {s_{15}}+ {s_{34}})}\\
&+\frac{15}{4 (- {s_{12}}- {s_{15}}+ {s_{34}})}+\frac{7  {s_{34}}}{4  {s_{23}} (- {s_{12}}- {s_{23}}+ {s_{45}})}-\frac{5  {s_{34}}}{4 ( {s_{45}}- {s_{12}}) (- {s_{12}}- {s_{23}}+ {s_{45}})}\\
&-\frac{15  {s_{34}}}{8  {s_{23}} (- {s_{12}}+ {s_{34}}+ {s_{45}})}+\frac{1}{2 (- {s_{12}}- {s_{23}}+ {s_{45}})}+\frac{4  {s_{34}}}{ {s_{12}}  {s_{45}}}-\frac{11  {s_{34}}}{4  {s_{45}} ( {s_{45}}- {s_{12}})}\\
&-\frac{15}{8 (- {s_{12}}+ {s_{34}}+ {s_{45}})}+\frac{5}{4 ( {s_{45}}- {s_{12}})}+\frac{4}{ {s_{12}}}-\frac{3  {s_{34}}^2}{ {s_{45}}^2 (- {s_{15}}+ {s_{23}}+ {s_{34}})}\\
&+\frac{ {s_{34}}}{4  {s_{45}} (- {s_{15}}+ {s_{23}}+ {s_{34}})}+\frac{ {s_{45}}}{2 ( {s_{34}}+ {s_{45}}) ( {s_{15}}- {s_{23}}+ {s_{45}})}+\frac{3  {s_{34}}}{ {s_{45}}^2}-\frac{1}{2 ( {s_{34}}+ {s_{45}})}-\frac{1}{4  {s_{45}}}
\end{aligned}
\end{align}
For this particular coefficient, we achieve a byte size reduction of a factor from 249464 bytes to 15192 bytes, which is 6\% of the original size. Besides the reduction of the byte-size, we see that the highest total degree of the numerator of $c_{107}$ also significantly decreased from 11 to 3.


\bibliographystyle{JHEP}
\bibliography{bibtex}

\end{document}